%% file: main.tex
\documentclass[11pt]{article}

\input{settings/packages}

\input{settings/macros}
\input{settings/environments}

\title{The Cost of Consistency:\\ Submodular Maximization with Constant Recourse}

\author{Paul D\"utting\thanks{Google Research, Email: \texttt{\{duetting,siviol,ashkannorouzi,zadim\}@google.com}}  \and Federico Fusco\thanks{Sapienza University of Rome, Email: \texttt{federico.fusco@uniroma1.it}} \and Silvio Lattanzi$^{*}$ \and Ashkan Norouzi-Fard$^{*}$ \and Ola Svensson\thanks{École Polytechnique Fédérale de Lausanne, Email: \texttt{ola.svensson@epfl.ch}} \and Morteza Zadimoghaddam$^{*}$}

\date{}

\begin{document}
\maketitle

\begin{abstract}
    \input{sections/00-abstract}
\end{abstract}
\pagenumbering{gobble}

\clearpage 

\pagenumbering{arabic}

\input{sections/10-introduction}

\input{sections/20-model}

\input{sections/11-overview}

\input{sections/30-meta_algorithm}

\input{sections/40-exp-time}

\input{sections/50-poly-time}

\section*{Acknowledgements}
    The authors thank David Wajc for pointing to additional related work.
    The work of Federico Fusco is supported by the MUR PRIN grant 2022EKNE5K “Learning in Markets and Society”, the PNRR Project: “SoBigData.it - Strengthening the Italian RI for Social Mining and Big Data Analytics” (Prot. IR0000013 - Avviso n. 3264 del 28/12/2021), and the FAIR (Future Artificial Intelligence Research) project PE0000013, spoke 5. Part of the work of Ola Svensson was done while he was a Visiting Faculty at Google Research; Ola Svensson's work is supported by the Swiss State Secretariat for Education, Research and Innovation (SERI) under contract number MB22.00054.
    
        \bibliographystyle{plainnat}
        \bibliography{references}

\clearpage
\appendix

\input{sections/100-appendix}

\end{document}

%% file: settings/packages.tex
\usepackage[utf8]{inputenc}
\usepackage[margin=1in]{geometry}
\usepackage[T1]{fontenc}
\usepackage{microtype}
\usepackage{algorithm}
\usepackage[noend]{algorithmic}
\usepackage{amsmath, amssymb, amsthm}
\usepackage{thmtools}
\usepackage{thm-restate}
\usepackage{blindtext}
\usepackage{hyperref} % Must load hyperref before cleverref
\usepackage{comment}
\usepackage{enumerate}
\usepackage{mathtools}
\usepackage{subcaption}
\usepackage[dvipsnames]{xcolor}
\usepackage{natbib}
\usepackage{makecell}

\usepackage{dsfont}
\usepackage{xspace,nicefrac}
\usepackage{multirow}
\usepackage{microtype}
\usepackage{graphicx}
\usepackage{booktabs} % for professional tables
\usepackage[symbol]{footmisc}

\usepackage{afterpage}

\usepackage{graphicx}
\usepackage{settings/color-edits}
\addauthor{pd}{orange} 
\addauthor{ff}{blue}
\addauthor{mz}{teal}
\addauthor{os}{red}
\addauthor{sl}{magenta}

\usepackage{cleveref}

\crefformat{equation}{(#2#1#3)}

%% file: settings/macros.tex
\newcommand\morteza[1]{\textcolor{olive}{Morteza: #1}}
\DeclareMathOperator*{\argmin}{arg\,min}

\newcommand{\E}[1]{\mathbb{E}\left[#1\right]}
\newcommand{\EO}[0]{\mathbb{E}}
\renewcommand{\P}[1]{\mathbb{P}\left(#1\right)}

\newcommand{\swapping}{\textsc{Swapping}\xspace}

\newcommand{\OPT}{\textnormal{OPT}}
\newcommand{\ALG}{\textnormal{ALG}\xspace}
\newcommand{\greedy}{\textsc{Greedy}\xspace}
\newcommand{\checkpoint}{\textsc{Check-Point}\xspace}

\newcommand{\swap}{\textsc{Any-Swap}\xspace}
\newcommand{\GreedyCertificate}{\textsc{Greedy-with-Certificate}\xspace}
\newcommand{\ExpoAlg}{\textsc{MinMaxSampling}\xspace}

\newcommand{\ind}[1]{\mathds{1}\left(#1\right)}
\newcommand{\e}{\varepsilon}

\newcommand{\local}{\textsc{Local-Search}\xspace}

\newcommand{\Sold}{S_{\text{old}}}
\newcommand{\Snew}{S_{\text{new}}}
\newcommand{\oldOPT}{\OPT_{\textsf{now}}}

\newcommand{\block}{\Delta}

\newcommand{\blockid}{\tau}
\newcommand{\GreedySol}{\mathrm{S}}
\newcommand{\AugmentedSol}{\mathrm{A^+}}
\newcommand{\FinalSol}{\mathrm{A}}
\renewcommand{\k}{\kappa}

\newcommand{\Vfuture}{V_{\textsf{future}}}
\newcommand{\Vnow}{V_{\textsf{now}}}
\newcommand{\ExtraSample}{\eta}
\newcommand{\threshold}{\gamma}

\newcommand{\DistributionA}{\mathcal{D}}

\newcommand{\cA}{\mathcal{A}} 
 
\newcommand{\cE}{\mathcal{E}}

%% file: settings/environments.tex
\theoremstyle{plain}
\newtheorem{theorem}{Theorem}[section]
\newtheorem{lemma}[theorem]{Lemma}

\newtheorem{claim}[theorem]{Claim}

\theoremstyle{definition}
\newtheorem{definition}[theorem]{Definition}
\newtheorem{example}[theorem]{Example}
\newtheorem{observation}[theorem]{Observation}

%% file: sections/00-abstract.tex
In this work, we study online submodular maximization, and how the requirement of maintaining a stable solution impacts the approximation. In particular, we seek bounds on the best-possible approximation ratio that is attainable when the algorithm is allowed to make at most a constant number of updates per step. 
We show a tight information-theoretic bound of $\nicefrac{2}{3}$ for general monotone submodular functions, and an improved (also tight) bound of $\nicefrac{3}{4}$ for coverage functions. Since both these bounds are attained by non poly-time algorithms, we also give a poly-time randomized algorithm that achieves a $0.51$-approximation.
Combined with an information-theoretic hardness of $\nicefrac{1}{2}$ for deterministic algorithms from prior work, our work thus shows a separation between deterministic and randomized algorithms, both information theoretically and for poly-time algorithms.

%% file: sections/10-introduction.tex
\section{Introduction}

A rapidly growing literature examines online combinatorial problems from a stability perspective. The premise is that in many online problems decisions are not irrevocable, but changing the solution comes at a cost. It is thus desirable to have online algorithms that are ``consistent'' or have ``small recourse'' \citep[e.g.,][]{BhattacharyaEtAl2024,fichtenberger2021consistent,GuptaGKKS22,GuptaL20,lattanzi2017consistent,lkacki2024fully}, in the sense that they make few updates per step or overall. 
In this work, we explore this requirement for the fundamental problem of submodular maximization with cardinality constraints, focusing on a stability notion that limits the allowed updates per step. We adopt the model of \cite{DuettingFLNZ24}, but augment it to allow for randomized algorithms.

The problem is as follows. A (possibly randomized) algorithm faces an adversarially created sequence of $n$ elements arriving one-by-one, on which a monotone submodular function $f$ is defined. Informally, the algorithm's goal is to maintain a ``high value'' subset of at most $k$ elements without changing it ``too much'' in each step. Towards defining this more formally, let $\ALG_t$ denote the algorithm's solution after the $t^{th}$ insertion. When a new element $e_t$ arrives, the algorithm can decide to update the solution from the previous step, $\ALG_{t-1}$, to a new solution $\ALG_t$. We strive for the following two properties:
\begin{itemize}
\item \textbf{Approximation:} We say that the algorithm is an \emph{$\alpha$-approximation}, for $\alpha \leq 1$, if at each time step it ensures that $\mathbb{E}[f(\ALG_t)] \geq \alpha f(\OPT_t)$ for all $t$, where $\OPT_t$ is the optimal solution among the sets of elements arrived thus far. 

\item \textbf{Consistency:} 
Let $C$ be a constant; an algorithm is \emph{$C$-consistent}, if, for any possible choice of the random bits, $|\ALG_t \setminus \ALG_{t-1}| \leq C$ for all $t$. An algorithm is \emph{consistent} if there exists a constant $C$ such that the algorithm is $C$-consistent.
\end{itemize}

In this work, we investigate the ``cost of consistency.'' In particular, we want to understand the fundamental information-theoretic boundaries on the quality of solutions maintained by a (possibly non poly-time) consistent algorithm.

\subsection{Our Contribution}

\paragraph{Tight Information-Theoretic Bound.} Our first main result is a \emph{tight} characterization of the information-theoretic hardness of the problem for (possibly randomized) algorithms.

    \medskip
    \noindent \textbf{Theorem.}~(\Cref{thm:exp-time,thm:hardness07}) There exists a consistent randomized (non-polynomial time) algorithm that is a $\nicefrac{2}{3}$-approximation, and this is information-theoretically tight.
    \medskip

This in particular shows that there is a strict separation between randomized and deterministic algorithms. In fact, as already observed in \cite{DuettingFLNZ24}, a simple construction shows that no consistent deterministic algorithm, whether poly-time or not, can achieve a better than $\nicefrac{1}{2}$-approximation. 
The strict separation between randomized and deterministic algorithms implied by our theorem is somewhat surprising as in many submodular optimization problems randomization does not help. Citing from \citet[][p.1]{BuchbinderFeldman2018}: 

\begin{quote}
``an interesting fundamental problem in this area is whether randomization is inherently necessary for obtaining good approximation ratios'' 
\end{quote}

Indeed, for many fundamental problems there is no gap, and there is an active line of work that seeks to close gaps where they exist \citep[e.g.,][]{BuchbinderFeldman2018,BuchbinderFeldman24}.

For the standard offline setting it is believed that all polynomial-time randomized algorithms can be derandomized; in particular, the celebrated result of~\cite{RussellW97} shows that, if SAT has no subexponential sized circuits, then $\mathsf{BPP} = \mathsf{P}$, i.e., all randomized polynomial-time algorithms can be derandomized. 
Consider the standard offline problem of maximizing a submodular function. In the unconstrained case, \cite{BuchbinderFeldman2018} show how derandomization yields a $\nicefrac{1}{2}$-approximation, matching the impossibility for randomized algorithms of \cite{FeigeMV11}. For monotone functions and cardinality constraints, the deterministic Greedy algorithm achieves a tight $(1-\nicefrac{1}{e})$-approximation \citep{NemhauserWF78,Feige98}. For the more general case of matroid constraints, randomized algorithms that achieve a $(1-\nicefrac{1}{e})$-approximation were known for more than a decade \citep{CalinescuCPV11}. In a recent break-through, \cite{BuchbinderFeldman24} showed that this bound can also be achieved with a deterministic algorithm. 

Most importantly, this no-gap phenomenon also extends to online settings. The most prominent online problem with no gap between deterministic and randomized algorithms is probably the streaming setting, which admits a deterministic $\nicefrac 12$ approximation that is tight even if randomization is allowed \citep{FeldmanNSZ23}.

\paragraph{Submodular vs.~Coverage Functions.}

We also show that there is a gap between general submodular functions and coverage functions.

    \medskip
    \noindent \textbf{Theorem.}~(\Cref{thm:exp-time-coverage,thm:3/4_randomized}) There is a consistent  randomized (non-polynomial time) algorithm that is a $\nicefrac{3}{4}$-approximation for coverage functions, and this is information-theoretically tight.
    \medskip

We find this separation between coverage functions and general submodular functions rather intriguing, as we are not aware of any other natural submodular optimization problem with such a gap.
Indeed, for several canonical problems in (monotone) submodular optimization, there is no gap between the two. For instance, the already mentioned hardness results of \cite{Feige98} and \cite{FeldmanNSZ23} are for coverage functions.
We note that the best-known information-theoretic hardness for consistent deterministic algorithms is based on coverage functions \citep{DuettingFLNZ24}, while the best-known deterministic algorithms achieve a $0.3818$-approximation for general monotone submodular functions. Proving (or disproving) that there is a gap between general submoudular functions and coverage functions in the deterministic setting is a fascinating open problem.

\begin{table}[t]
\centering
\begin{tabular}{@{}l|cc@{}}
\toprule
     & poly time &  info-theoretic\\
\midrule
det. & \makecell{$(0.3818,\nicefrac{1}{2}]$\\{\small \cite{DuettingFLNZ24}}} & \makecell{$(0.3818,\nicefrac{1}{2}]$\\{\small \cite{DuettingFLNZ24}}} \\
\midrule
rand. & \makecell{$(\mathbf{0.51},1-\nicefrac{1}{e}]$ \\ {\small Thm.~\ref{thm:poly-time} \& \cite{Feige98}}} & \makecell{$\nicefrac{\mathbf{3}}{\mathbf{4}}$ ({\small coverage)} and $\nicefrac{\mathbf{2}}{\mathbf{3}}$ \small{(submodular)}\\ {\small Thms.~\ref{thm:exp-time-coverage} and \ref{thm:3/4_randomized} \& Thms.~\ref{thm:exp-time} and \ref{thm:hardness07}}}\\
\bottomrule
\end{tabular}
\caption{Overview of results: Lower bounds (pos.) and upper bounds (neg.) on the approximation ratio of consistent algorithms. All bounds are for monotone functions, and apply to both coverage functions and general submodular functions, except where indicated.}
\label{tab:overview}
\end{table}

\paragraph{Separation for Efficient Algorithms.} 
As our final result we demonstrate that it is even possible to break the (information-theoretic) barrier for deterministic algorithms with an efficient poly-time randomized algorithm. 
This shows that there is a gap between deterministic and randomized algorithms also when we insist on poly time.

\medskip
    
\noindent \textbf{Theorem.}~(\Cref{thm:poly-time}) There exists a consistent poly-time randomized algorithm that is a $0.51$-approximation.
    
\medskip

Although this result is not tight, we note that, under the assumption that $\mathsf{P} \neq \mathsf{NP}$, we cannot hope to get a better than $(1-\nicefrac{1}{e})\approx 0.64$-approximation with a poly-time algorithm (whether consistent or not). We leave as an intriguing open problem, whether there is a ``cost of consistency'' for poly-time randomized algorithms. We summarize our results in  Table~\ref{tab:overview}.

\medskip 
\noindent\textbf{Discussion.}
We conclude with a brief discussion of our modeling choices regarding consistency. In our eyes, this question has two main dimensions:
\begin{enumerate}
\setlength\itemsep{0pt}
    \item[(i)] Should consistency be enforced per step (as we do here), or should we aim for low consistency in an amortized sense?
    \item[(ii)] If we enforce it per step, why a constant number of changes per step, as opposed to say at most a single change per step?
\end{enumerate}  

For (i) we believe that the ``per step'' requirement is the stronger and hence perhaps the more appealing notion for practical applications, but it is also possible to imagine applications where an amortized version makes sense. In this work we make progress on the former notion, and leave the latter for future work.\footnote{We note that the line of work on fully dynamic submodular maximization is closely related to the amortized notion of consistency. See discussion in \Cref{sec:related_work}.}
For (ii) we believe that allowing for a constant number of changes (although asymptotically equivalent to let's say one or two changes) gives enough leverage to uncover the algorithmic complexity of the problem. In fact, as we will see,  it enables a reduction to a very clean two-stage stochastic optimization problem (which we dub ``addition-robust submodular maximization''), which nicely identifies and isolates the core challenge.

\subsection{Further Related Work} \label{sec:related_work}

The work most closely related to ours is \citet{DuettingFLNZ24}, where the authors study consistent submodular maximization for deterministic algorithms. Beyond the already mentioned hardness of $\nicefrac 12$, they also provide a $1$-consistent $0.3178$-approximation and a $\nicefrac{1}{\e}$-consistent $(0.3818 - \e)$-approximation.

Prior work, while not explicitly studying consistency, also implies consistent algorithms. In particular, the \swapping algorithm of \cite{ChakrabartiK14}---designed for the streaming version of the problem---is $1$-consistent and provides a $\nicefrac 14$-approximation. We remark that there are crucial differences between the streaming model and our problem: In the streaming model, the algorithm is not able to go back to previous elements, unless it decided to store those in memory. On the other hand, the algorithm is not constrained by the requirement that the solution should not change too much between any two consecutive time steps.

% \paragraph{The Streaming Model.} A closely related stream of work considers (online) submodular maximization in the streaming model (as pioneered by \cite{ChakrabartiK14} and \cite{ChekuriGQ15}). In this classic model, just as in our problem, elements arrive one-by-one, and the algorithm is required to be able to produce a solution of high value at any point in time. 

% However, the main challenge in the streaming model is quite different from the one in our model: In the streaming model, the algorithm is not able to go back to previous elements, unless it decided to store those in memory. On the other hand, the algorithm is not constrained by the requirement that the solution should not change too much between any two consecutive time steps. %to make only a constant number of changes to the solution at each time step.

% An important algorithm for the streaming problem, known as the \swapping algorithm (\cite{ChakrabartiK14}), uses no memory at all, and sometimes swaps in the arriving element. This algorithm 
% provides a $\nicefrac{1}{4}$ approximation, and is $1$-consistent. However, it can also easily be seen that the approximation guarantee of \swapping is tight (\cite{DuettingFLNZ24}, Appendix A).

%\paragraph{Online Submodular Maximization With Preemption.} 
Another closely related line of work concerns online submodular maximization with preemption (a.k.a.~online submodular maximization with free disposals)  \citep{BuchbinderFS15a,Chan0JKT17}. The model allows the algorithm to drop previously accepted elements, and replace them with the one that just arrived. As a benchmark, they also consider the best solution among the elements arrived thus far. This naturally leads to $1$-consistent algorithms. Unlike in our model, however, the algorithm cannot go back to previously arrived elements.

For the cardinality constrained problem, with monotone submodular objective function, \cite{BuchbinderFS15a} %observe that the \swapping algorithm yields a $\nicefrac 14$ approximation.
%there exists a deterministic, poly-time $\nicefrac{1}{4}$ approximation  \citep{Badanidiyuru11,ChakrabartiK14}. 
%On the hardness side, they 
show a $\nicefrac{1}{2}+\e$ impossibility for deterministic algorithms, and a $\nicefrac{3}{4}+\e$ impossibility for randomized algorithms. Our $\nicefrac{2}{3}$ hardness result improves on the latter, as we consider the same benchmark but are giving more power to the algorithm (ability to go back to previously arrived elements, and constant number of changes).
% The hard instances in their paper, however, use submodular functions that are not coverage functions. Another important difference to our model (and why the previous hardness results don't apply out-of-the-box) is that in our case, the algorithm %may make up to $O(1)$ changes to the solution, and 
% is allowed to bring back (up to a constant number of) items it previously skipped. 
% \ola{I will refrain from saying too much that those hardness results are not applicable; We have stronger results and it may be that they can be easily adapted by using the lifting technique that we use. In any case, we get stronger results so we don't need to argue that they don't apply too much.}\pdcomment{Fine by me.  Shall we just drop everything from ``The hard instances in their paper, however,...'' onwards?}
%
%The paper also gives additional results on other feasibility constraints and non-monotone submodular objective functions.
%
%Another closely related work is a paper by \cite{Chan0JKT17}. They show that for $k$-uniform matroids (a.k.a.~cardinality constraints), there is 
On the positive side, \cite{Chan0JKT17} give a deterministic algorithm with competitive ratio at least $0.2959$, with the ratio tending to $\nicefrac{1}{\alpha_\infty} \approx 0.3178$ as $k$ approaches infinity. Note, their approximation ratio matches the one of the $1$-consistent algorithm of \citet{DuettingFLNZ24}, which in turn, also fits in the preemption model.

Additionally, \cite{Chan0JKT17} also study the problem with matroid constraints, which generalize cardinality constraints. In particular, they prove a hardness result of $\nicefrac{1}{4}$ for deterministic algorithms, that holds for the class of partition matroids, thus showing that the \swapping algorithm (which naturally fits into the preemption framework) is optimal. They also design a randomized algorithm for partition matroids, characterized by a competitive ratio of $0.3178$. They thus show that randomization allows to break the deterministic barrier for partition matroids. Our general finding is qualitatively similar. However, we show it for a different problem, and already for cardinality constraints.

Finally, our work is also related to the work on fully dynamic submodular maximization \citep{LattanziMNTZ20,ChenP22,BanihashemBGHJM24}. 
There the algorithmic goal is to maintain a good solution with small amortized update time and thus small \emph{amortized} consistency. For instance, \citet{LattanziMNTZ20}  provide a $\nicefrac{1}{2}$-approximation that is characterized by a poly-logarithmic amortized update time and, thus, a poly-logarithmic amortized consistency. We point out that such algorithms do not enforce our notion of ``per-step'' consistency as they may (rarely) recompute the whole solution from scratch.

%% file: sections/20-model.tex
\section{Preliminaries}
\label{sec:preliminaries}

We study the following problem. 
There is a ground set of elements, $X$, of cardinality $n$, and a (monotone submodular) set function $f:2^X \rightarrow \mathbb{R}_{\ge 0}$. Elements arrive one-by-one, in discrete time steps $t = 1, \ldots, n$.
We use $x_t \in X$ to refer to the element that arrives at time step $t$; and $X_t =\{x_1, \ldots, x_t\}$ to refer to the set of elements that arrive in the first $t$ time steps. 
% \ffcomment{In Section 3 and 5, we denote with $\ALG_t$ the solution maintained by the algorithm, to avoid confusion with $S$, which is the output of Greedy}
%%
We seek to design an algorithm that maintains a solution $\ALG_t \subseteq X_t$ of cardinality at most $k$ in a dynamic way. If the algorithm is randomized, then $\ALG_t$ is a random subset of $X_t$. 
We assume that the algorithm does \emph{not} know $n$ ahead of time, and that at any given point in time $t$, the algorithm can access $f$ on the set of elements $X_t$ that have arrived thus far through a value oracle. A \emph{value oracle} is given a set $S$ and returns $f(S)$. When evaluating the running time of an algorithm, we assume that each value query takes constant time.

Our goal is to have an approximately optimal set $\ALG_t$ at all times $t$ (a set of high value, relative to the optimal set of size $k$ at any given point, as given by a monotone submodular function $f$),
while making only a constant number of changes to the solution at every step. 
%%
%Unlike in the classic streaming model, we assume that, at all times, the algorithm has access to all elements that have arrived so far. Instead, we will impose a limit on how many changes the algorithm can make to the solution $S_t$ at each time step $t$.
%
%\paragraph{Adversary model.} %We consider two types of adversaries. 
We assume that the input is generated adversarially. We consider an \emph{oblivious adversary}, which knows the (possibly randomized) algorithm (but not the outcome of its internal random bits), and decides the input in advance.

\paragraph{Submodular Functions.}
The value of a set of elements is given by a set function $f: 2^X \rightarrow \mathbb{R}_{\geq 0}$. Given two sets $S, T \subseteq X$, the \emph{marginal gain} of $S$ with respect to $T$, $f(S \mid T)$, is defined as
\[
    f(S \mid T) = f(S \cup T) - f(T).
\]
When $S$ is the singleton $\{x\}$, we use the shorthand $f(x \mid T)$ rather than $f(\{x\} \mid T)$. Function $f$ is  \emph{monotone} if $f(x \mid T) \geq 0$ for every element $x \in X$ and every set $T \subseteq X$; and  \emph{submodular} if, for any two sets $S \subseteq T \subseteq X$, and any element $x \in X \setminus T$ it holds that
\(
    f(x \mid S) \geq f(x \mid T).
\)

%When evaluating the running time of algorithms, we assume that we have \emph{oracle access} to $f$. Such an oracle, receives any set $S \subseteq X$ as input, and returns $f(S)$. 

%We consider the \emph{consistent submodular maximization problem} 
%%(subject to a cardinality constraint) 
%introduced in \cite{DuettingFLNZ24}. An algorithm for this problem, observes the input stream $x_1, x_2, \ldots, x_n$ one-by-one, and maintains a set $S_t \subseteq X_t$ of size at most $k$ in a dynamic way. We consider both deterministic and randomized algorithms. If the algorithm is randomized, then $S_t$ is a random subset of $X_t$. 

\paragraph{Coverage Functions.} 
%In many optimization problems, we seek to find a solution that covers a certain set of items. 
%We are also interested in the 
Given a ground set $Y$ and a set of elements $X$ that are subsets of $Y$, a coverage function $f: 2^{X} \rightarrow \mathbb{R}_{\geq 0}$ maps a subcollection of sets to a non-negative real value representing the total ``coverage'' achieved by these sets, for any $S \subseteq X$.  More precisely, for any $S \subseteq X$, $f(S) = |\cup_{x \in S} x|$. 
Coverage functions are monotone and submodular.

\paragraph{Approximation Guarantee.} 

We seek algorithms that provide a good approximation guarantee to the best solution at each time step. %be competitive with the best solution at each time step. 
Formally, let $\OPT_t \subseteq X_t$ be the solution of maximum value (according to $f$), among all solutions of size at most $k$ among the elements that have arrived so far. We say that the algorithm is an \emph{$\alpha$-approximation}, for $\alpha \leq 1$, if $
    \mathbb{E}[f(\ALG_t)] \geq \alpha \cdot f(\OPT_t)$, {for all $t$.}

\paragraph{Consistency.}
In addition, we require that the algorithm is ``consistent'' in that it makes at most a constant number of changes to the solution $\ALG_t$ at each time step. We only impose this constraint on the algorithm, not on the benchmark (the optimal solution $\OPT_t$).

Formally, let $n_t = |\ALG_t \setminus \ALG_{t-1}|$ be the number of changes in the solution, going from $\ALG_{t-1}$ to $\ALG_t$.\footnote{We note that we could also define consistency using the symmetric difference. We prefer this notion because it's simpler, and basically equivalent. It is also how consistency was defined in prior work \citep{DuettingFLNZ24}.} %\pdcomment{Added FN.} 
Note that for a randomized algorithm, $n_t$ is a random variable. Fix a constant $C$. We say that a (possibly randomized) algorithm is \emph{$C$-consistent}, if, for all choices of the random bits, $n_t \leq C$ for all $t$. We say that it is \emph{consistent}, if there exists a constant $C$ such that it is $C$-consistent.

%% file: sections/11-overview.tex
\section{Techniques and Ideas}

This section outlines the techniques and ideas employed in this work.
In \Cref{subsec:reduction}, we describe a simple but crucial reduction from (online) consistent submodular maximization to a two-phase problem, that we name addition-robust submodular maximization. Then, in \Cref{subsec:information_tight} we sketch the main ideas to derive tight information-theoretic bounds on the problem. Finally, \Cref{subsec:poly-time} is devoted to presenting the ideas underlying our randomized polynomial time algorithm that guarantees a better-than-$\nicefrac 12$ approximation.

\subsection{Addition-Robust Submodular Maximization.}
\label{subsec:reduction}

    One of the main difficulties arising from maintaining a consistent solution resides in the fact that the  dynamic optimum may change quickly -- even completely after a single insertion -- while a consistent algorithm may need $\Theta(k)$ time steps to move from one given solution to another, with the extra-complication of maintaining a {good} quality solution \emph{throughout} the transition.
    
    In Section~\ref{sec:addition-robust}, {we present a simple meta-algorithm, \checkpoint, that reduces the consistent submodular maximization problem to an offline problem, that we call addition-robust submodular maximization. The idea is as follows:} given a precision parameter $\e$, \checkpoint divides the time horizon into $\nicefrac{1}{(\e k)}$ contiguous blocks of $\e k$ insertions and interpolates -- within each block -- between two ``check-point'' solutions, while adding all the recently inserted elements. This interpolation is randomized: a contiguous sub-block of length $\e^2 k$ time steps is chosen uniformly at random and hosts the transition from two consecutive  ``check-point'' solutions. In this way, each fixed time step has low probability (exactly $\e$) of entailing a change of $\nicefrac{1}{\e^2}$ elements in the solution (that might worsen its quality). Note, since at each time step at most $O(\nicefrac 1{\e^2})$ elements are replaced in the solution, \checkpoint is consistent. For this meta-algorithm to maintain a good solution overall, it is crucial that the check-point solutions have high value and are robust with respect to the new elements that may arrive during the following block; stated differently, we would like to avoid that the newly inserted elements have low marginal contribution to the check-point solution, while having high contribution with respect to other sets.
    
    This desirable property is captured and made precise by the addition-robust submodular maximization problem, that is defined as follows. Let $V$ be the domain of a monotone submodular function $f$, and assume that $V$ is partitioned into current elements $V_{\textsf{now}}$ and future elements $V_{\textsf{future}}$. The goal is to design a (possibly randomized) algorithm $\cA$ for this problem that is $\alpha$-addition robust, as per the following definition.
    % \textcolor{red}{[Paul: I think we should comment on why it is $\kappa$ rather than $k$ in Definition 3.1.]}

    \begin{definition} 
    \label{def:insertion-robust} 
        Let $\alpha \leq 1$. We say that algorithm $\cA$ is \emph{$\alpha$-addition robust} for submodular maximization subject to a cardinality constraint $\kappa$,\footnote{{We use $\kappa$ instead of $k$ when we refer to cardinality constraint for the addition-robust problem to avoid ambiguity: for the reduction to work, the checkpoint solutions contain $\kappa =(1-O(\e))k$ elements, so that there is room for the $\e k$ new elements that arrive during the contiguous block of insertions.}} if it outputs a (possibly randomized) set $\cA[V_{\textsf{now}}]$ of cardinality at most $\kappa$ such that the following inequality holds for any $R \subseteq V_{\textsf{future}}$:
    \[
                    \E{f(\cA[V_{\textsf{now}}] \cup R)} \ge \alpha \cdot \max_{\substack{S^* \subseteq V_{\textsf{now}} \cup R \\ |S^*| \le \kappa}} f(S^*).
    \]
    \end{definition}
    \begin{figure}
    \centering
    \begin{subfigure}{0.4\textwidth}
        \includegraphics[width=\textwidth]{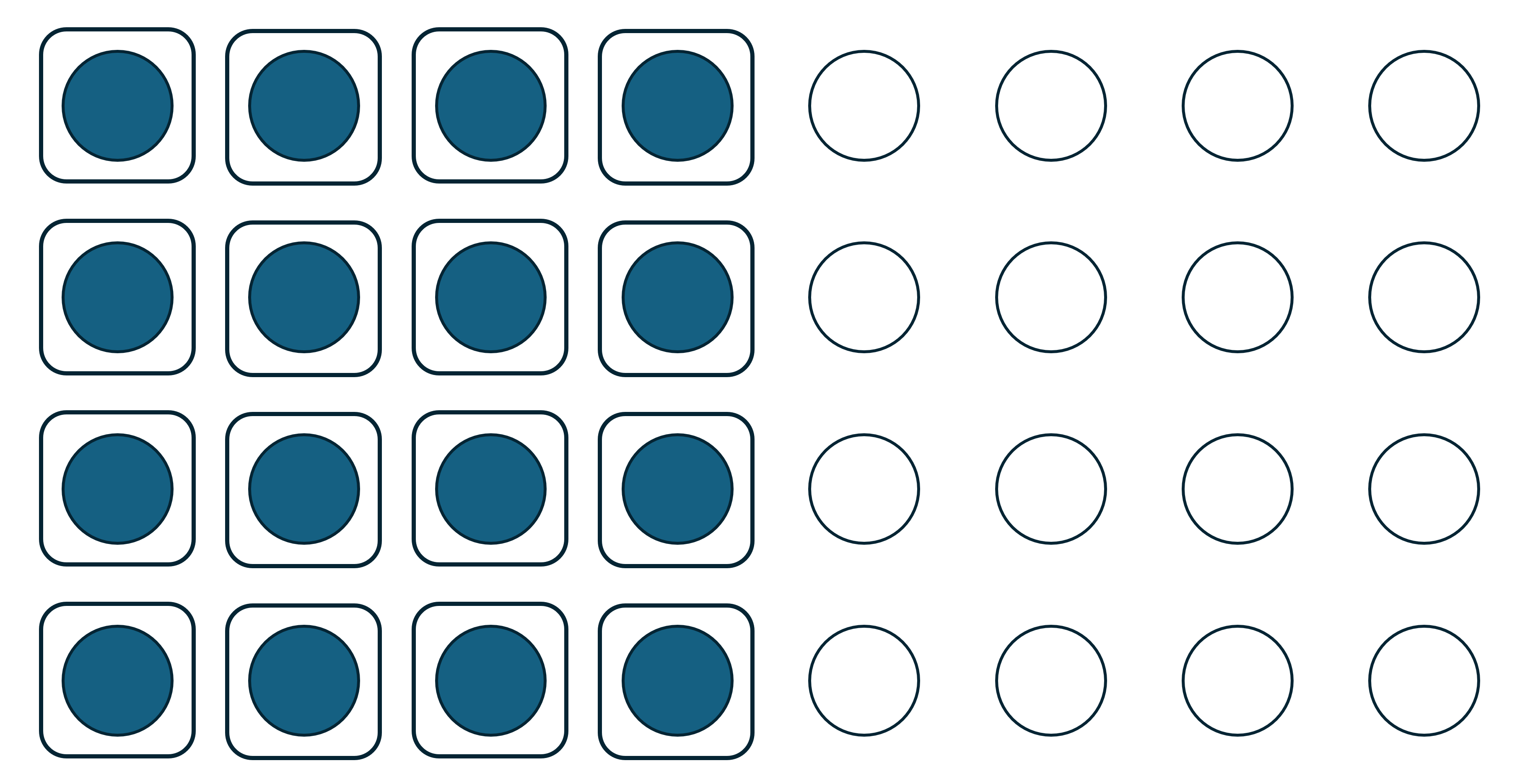}
        \caption{Elements covered by $\cA[V_{\textsf{now}}]$.}
        % \label{fig:first}
    \end{subfigure}
    \hfill
    \begin{subfigure}{0.4\textwidth}
        \includegraphics[width=\textwidth]{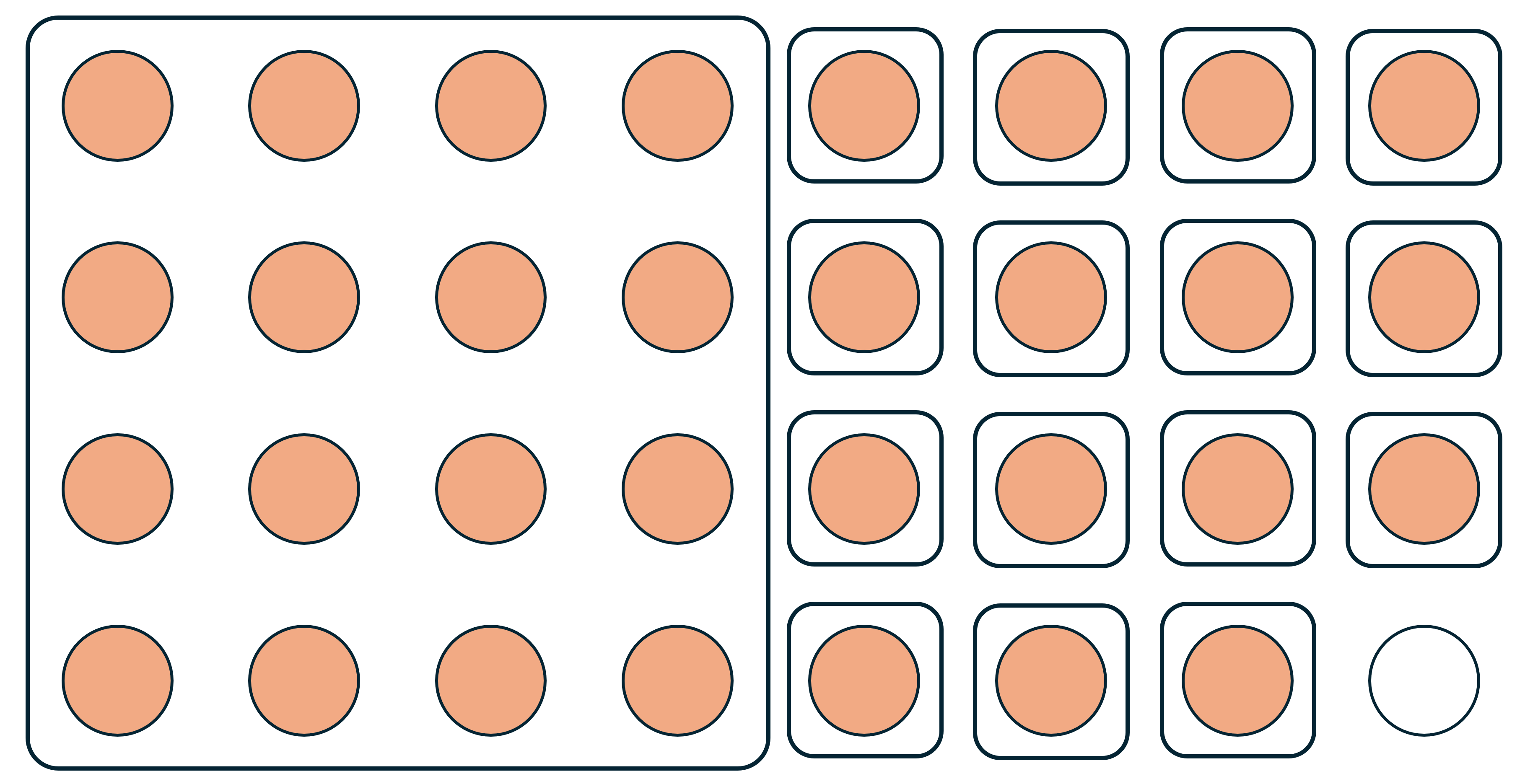}
        \caption{Optimal solution after inserting $R$.}
        % \label{fig:second}
    \end{subfigure}
    \caption{Visualization of the perfect alignment phenomenon described in \Cref{ex:alignment}.}
    \label{fig:alignment}
    \end{figure}
Intuitively, an addition-robust algorithm avoids situations in which the addition of some elements $R$ does not add any value to the current solution $\cA[V_{\textsf{now}}]$, while increasing dramatically the value of another subset of elements. 
In \Cref{sec:addition-robust}, we formally describe \checkpoint and prove that any efficient $\alpha$-addition robust algorithm can be plugged into \checkpoint to obtain an efficient consistent algorithm with roughly the same approximation factor $\alpha$. 
 \begin{restatable}{theorem}{reduction}
        \label{thm:reduction}
        Suppose $\cA$ is an $\alpha$-addition-robust algorithm. Then %meta-algorithm 
        \checkpoint with precision parameter $\e \in (0,1)$ is $O(\nicefrac{1}{\e^2})$-consistent and provides an $(\alpha-O(\e))$-approximation.
        \end{restatable} 
        We highlight that randomization is crucial in two steps of our approach: first, it is essential in the choice of the random interpolation times in \checkpoint and, second, it is necessary to get a better than $\nicefrac 12$ approximation for the addition-robust problem. In fact, it is easy to construct instances where no deterministic addition robust algorithm can achieve better $\nicefrac 12$ approximation, as the following example demonstrates.
        \begin{example}[Perfect Alignment]
        \label{ex:alignment}
            Consider a coverage function $g$ defined over a universe $U$ of $N$ elements $u_1, u_2, \dots, u_N$. The domain of the coverage function $g$ is given by the union of $\Vnow$ that contains all the singletons of $U$ and of $\Vfuture$ that contains all the subsets of $U$ of cardinality $\kappa$. Every fixed set $\cA[V_{\textsf{now}}] \subseteq \Vnow$ of cardinality $\kappa$ has value $g(\cA[V_{\textsf{now}}]) = \kappa$, however if $R$ is equal to the set of cardinality $\kappa$ that covers \emph{exactly} the same elements as $\cA[V_{\textsf{now}}]$, then $g(\cA[V_{\textsf{now}}] \cup R)=g(\cA[V_{\textsf{now}}]) = \kappa$, while the best solution in $\Vnow \cup R$ has always value exactly $2 \kappa -1$. Since $\kappa$ (and $N$) can be arbitrarily large, this example shows that no deterministic addition-robust algorithm can achieve a better than $\nicefrac 12$ approximation. We refer to \Cref{fig:alignment} for a visualization of the case where $N = 32$ and $\kappa = 16$.
        \end{example}
        
        Actually, as we detail in \Cref{app:local}, the standard local-search algorithm by \citet{NemhauserWF78} is a deterministic $\nicefrac 12$-addition-robust algorithm, thus settling the addition robust problem for deterministic algorithms. We next discuss two ways to circumvent this barrier through randomization.

\subsection{Tight-Information Theoretic Bounds}
\label{subsec:information_tight}
    
    We outline the key ideas and techniques employed to achieve tight information-theoretic results. For the lower bounds (positive results), we design a randomized algorithm, called \ExpoAlg, which exhibits the appropriate level of $\alpha$-addition robustness (and enables the design of consistent algorithms with the same approximation, as shown in \Cref{thm:reduction}). 
    % This property, as shown in \Cref{thm:reduction}, enables the development of a consistent submodular maximization algorithm. 
    The upper bounds (negative results) are based on carefully crafted hard instances that are informed by the insights gained in the analysis of \ExpoAlg.

    % We focus on describing the ideas underlying \ExpoAlg (see \Cref{sec:minmaxalg}). 
    A central challenge in designing an addition-robust algorithm is the unknown nature of the future set $R$.
    %To understand the gist of our approach, let's make the simplifying assumption that 
    To address this, consider the case where 
    we have a finite set of possible future scenarios $\mathcal{R} = \{R_1, \ldots, R_m\}$, one of which will occur. 
    In this setting, our objective is to find a solution that maximizes the expected value of the submodular function $f$, taking into account these potential future scenarios. To design an algorithm for this purpose, let $A_1, \ldots, A_m$ represent the optimal solutions among the current elements ($\Vnow$) for each of these scenarios, % \fede{so that if $R_i$ is the realized scenario, then the optimal solution in $\Vnow \cup R_i$ (i.e., the right-hand-side term in the definition of deletion robustness) is at most}
    so that if $R_i$ is the realized scenario, then the optimal solution is
%   . This means that, if the scenario that happens in the future is $R_i$, then the value of the optimum solution is at most 
$f(A_i \cup R_i)$.
%These sets are easily computable since we are assuming knowledge of the future scenario. 
    With this setup, the addition robust task reduces to finding 
    % primary challenge lies in finding 
    a solution $A$, that performs well across all possible future scenarios. Our approach consists in sampling our solution from a distribution $\DistributionA$ that hedges over all possible feasible solutions:
    % To address this, we develop an algorithm that constructs a distribution, $\DistributionA$, over all feasible solutions. Our solution is then obtained by sampling from this distribution. 
    % The objective is to identify a distribution that maximizes the minimum value attained across all future possibilities. We achieve an $\alpha$-approximate solution if

\begin{equation}
    \label{eq:hedge}
\EO_{A\sim \DistributionA}[f(A \cup R_i)] \geq \alpha \cdot f(A_i \cup R_i) \text{ for any } i=1,\ldots, m.
\end{equation}
%\textcolor{red}{[Paul: What do we mean by ``To that end''? I think it would be clearer to start a new paragraph and then say ``We formulate the problem of finding a distribution  that attains the best possible $\alpha$ as a the following  linear-program'' (or something like that)]} 
Note, the above condition implies the one appearing in the definition of addition robustness. The above idea can be formalized via a linear program that looks for the largest $\alpha$ for which a distribution respecting the inequalities in Formula~\ref{eq:hedge} exists:
\begin{align*}
    \text{Maximize} & \quad \alpha \\
    \sum_{A \subseteq V_{\textsf{now}}: |A| \leq \kappa} \lambda_A \cdot f(A \cup R_i) & \geq \alpha \cdot  f(A_i \cup R_i) \qquad \mbox{for $1 \leq i \leq m$}\\[-0.3cm]
    \sum_{A \subseteq V_{\textsf{now}}: |A| \leq \kappa} \lambda_A & \leq 1 \\[-0.2cm]
    \lambda_A & \geq 0 \qquad \mbox{for all $A \subseteq V_\textsf{now}: |A| \leq \kappa$.}
\end{align*}

Intuitively, our linear program considers all the scenarios in the future and assigns a probability $\lambda_A$ to each feasible solution so that if we sample with those probabilities we achieve an $\alpha$-approximate solution as in the inequalities in Formula~\ref{eq:hedge}. What remains is to bound the value of $\alpha$; to this end, we employ linear programming duality and careful simplifications, to 
%design a primal-dual approach that shows the problem of consistent submodular maximization 
show that this problem is equivalent to understanding the following question:
\begin{quote}
    \textbf{Question:} What is the largest $\alpha$ so that, for any sets of elements $A_1,\ldots, A_m \subseteq \Vnow$ and $m$ potential future subsets $R_1, R_2, \ldots, R_m$ it holds that
    % \textcolor{red}{[Paul: Is it clear enough that $i\sim [m]$ means uniformly at random?]}
\[
      \EO_{i,j \sim [m]} [f(A_i \cup R_j)]\footnote{We adopt the notation $i,j\sim [m]$ to denote that indices $i$ and $j$ are sampled independently and uniformly at random from $\{1,\dots,m\}$.} \geq \alpha \cdot \EO_{i \sim [m]} [f(A_i \cup R_i)]?
\]
\end{quote}

We note that this is in essence a \emph{correlation gap} question as we are comparing the value of $f$ when sampling the sets $A_i$ and $R_i$ jointly (on the right-hand-side) to the value obtained when sampling them independently (on the left-hand-side).%
\footnote{We remark that there are several different notions of correlation gap. A closely related one concerns the multi-linear extension of a set function $f$. For monotone submodular functions this gap is known to be $(1-\nicefrac{1}{e})$, and this is one of the main building blocks for achieving a $(1-\nicefrac{1}{e})$-approximation algorithm for maximizing a  monotone submodular function subject to a matroid constraint~\citep{CalinescuCPV11}.}
We also note that the actual algorithm (and its analysis) is more complicated, as we need to deal with the fact that the algorithm does not know $\mathcal{R}$ (in fact it does not even know $\Vfuture$ nor $|\Vfuture|$).

\paragraph{Coverage Functions.} We first discuss the simpler special case when $f$ is a coverage function. Let the underlying universe of $f$ be $Y$, i.e., every element $x$ corresponds to a subset of $Y$, and $f(S) = |\cup_{x\in S} x|$. Now consider any fixed $y \in Y$ and let $\beta$ and $\theta$  be the fraction of the $A_i$ and $R_i$ sets that contain $y$ for $1 \leq i \leq m$, respectively. If we sample two indices $i$ and $j$ independently and uniformly at random in $\{1,2, \dots, m\}$, we have that the probability that $y$ is covered by $A_i$ is $\beta$, while the probability that it is covered by $R_j$ is $\theta$. The indices $i$ and $j$ are sampled independently, thus the probability that $y$ is covered by both $A_i$ and $R_j$ is exactly $\theta \cdot \beta$. All in all, we have
%\textcolor{red}{[Paul: The following is easy. Should we add explanations nonetheless?]}
% \textcolor{red}{[Ashkan: I think we can use probability instead of expectation.]}\Fede:good point!
\begin{align*}
      &\mathbb{P}_{i,j \sim [m]} (y \in A_i \cup R_j)  = \beta + \theta - \beta \cdot \theta, \text{ and } \mathbb{P}_{i \sim [m]} (y \in A_i \cup R_i)  \leq \min\{1, \beta + \theta\}, 
\end{align*}
where the inequality follows by union bound. 
Combining these two equations, we get:
\begin{gather}
    \frac{\mathbb{P}_{i,j \sim [m]} (y \in A_i \cup R_j)}{\mathbb{P}_{i \sim [m]} (y \in A_i \cup R_i)} \geq  \frac{\beta + \theta - \beta \cdot \theta}{\min\{1, \beta + \theta\}} \geq \frac{3}{4}, \nonumber
\end{gather}
where the second quality holds for any $\beta$ and $\theta$ in $[0,1]$. Summing up the above inequality for all $y \in Y$ results in the desired inequality:
\begin{gather}
    \EO_{i,j \sim [m]} [f(A_i \cup R_j)] \geq \tfrac{3}{4} \cdot \EO_{i \sim [m]} [f(A_i \cup R_i)]. \nonumber
\end{gather}

The construction of the matching impossibility result is rather straightforward for coverage functions, and is based on the perfect alignment phenomenon, described in \Cref{ex:alignment}. The following observation shows that the $\nicefrac{3}{4}$ bound is tight for addition robustness. In \Cref{sec:hardness-coverage}, we turn this into a $\nicefrac{3}{4}$ impossibility result of the actual problem.

%\textcolor{red}{[Paul: Mmmh after reading the above I'm wondering how to best present/order the discussion here. One possibility would be to start with coverage, and say thus is the much easier result. Then after Observation 3.5 discuss the general submodular case?]}

\begin{observation}[Tight $\nicefrac 34$ instance for coverage]
    Consider the coverage instance over the $N$ elements universe with cardinality constraint $\kappa$ described in \Cref{ex:alignment}. The best distribution over $\Vnow$ for the addition robust problem is the uniform distribution over singletons. This means that the randomized set $\cA[V_{\textsf{now}}]$ covers exactly $\kappa$ elements of the universe, chosen uniformly at random (we refer to \Cref{fig:info_theoretic} for a visualization). Now, regardless of the choice of the subset $R$ by the adversary, its expected intersection with $\cA[V_{\textsf{now}}]$ has cardinality exactly $\nicefrac{\kappa}{N}.$ In particular, if $N = 2 \kappa$, this means that the gap between the optimal solution (that covers exactly $2\kappa-1$ elements), and $\cA[V_{\textsf{now}}] \cup R$ (that covers, in expectation, $\nicefrac 32 \kappa$ elements) converges to $\nicefrac 34$ as $\kappa$ grows to infinity.
\end{observation}

%\textcolor{red}{[Paul: I think 3.2 would benefit from another pass. I think it would be especially great (if possible) to be a bit more specific about the two steps that currently read like "now some hard and involved magic happens" (the step from LP to correlation gap, and where the 1/2 comes from for general submodular and why it's harder than coverage.]}

\paragraph{General Submodular Functions.} The case of general (monotone) submodular functions is more complex as we do not have as much structure as in the the case of a coverage function. {Initially, one might conjecture that any monotone submodular function can be transformed into a coverage function with an exponential increase in the instance size. However,  this is not the case.}
 In particular, no proof can achieve the same result for general submodular functions, as we show that -- in that case -- the tight bound is $\nicefrac 23$. 
The upper bound (impossibility) is informed by solving a linear program that solves the correlation gap for small instances (see \Cref{sec:hardness-submodular}). The structure of those instances then allows to highlight the differences between coverage and general submodular functions, and guide us towards useful properties (inequalities) that we use in the proof of the lower bound positive result (see \Cref{sec:minmaxalg}).

    \begin{figure}
    \centering
    \begin{subfigure}{0.4\textwidth}
        \includegraphics[width=\textwidth]{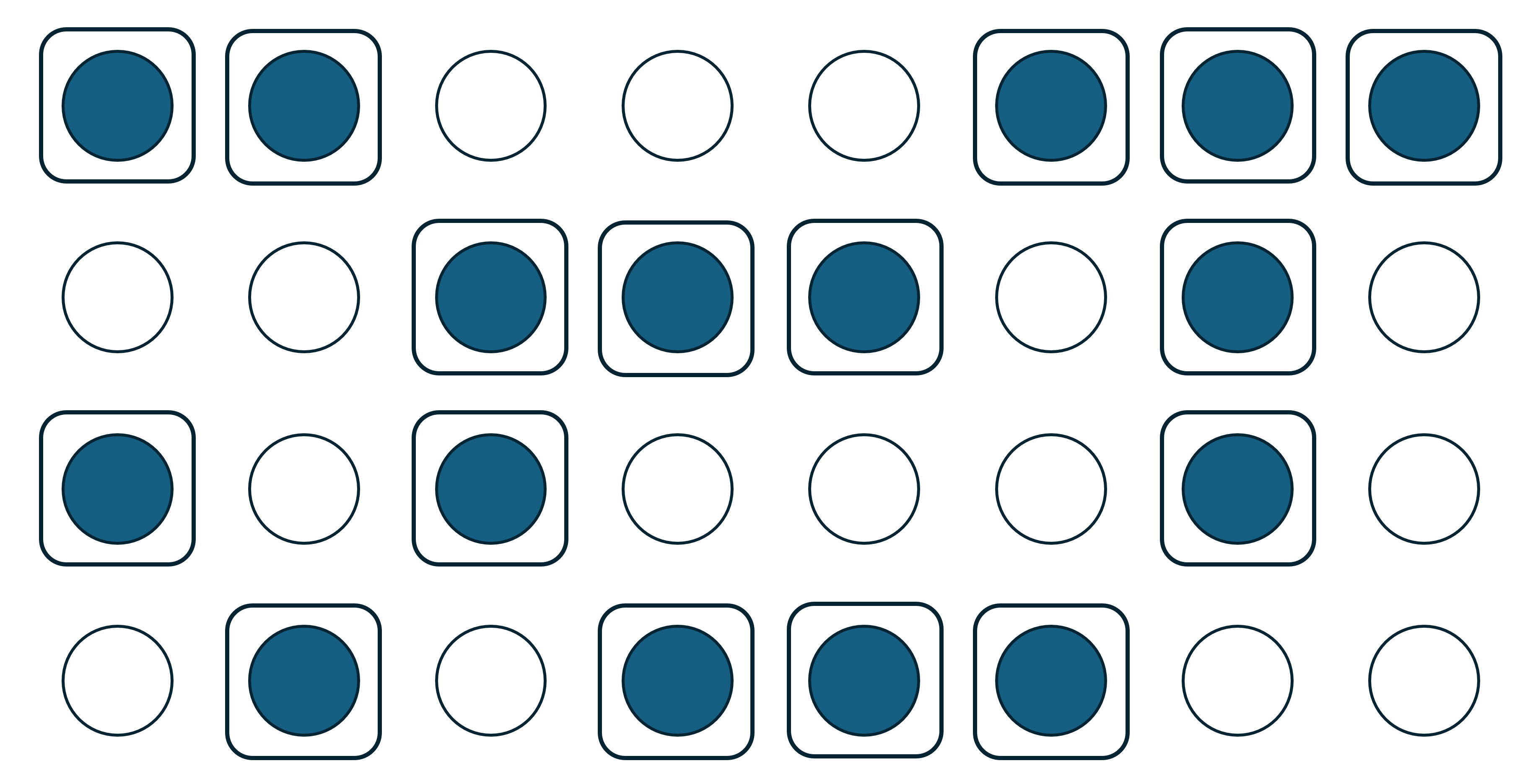}
        \caption{ $\cA[V_{\textsf{now}}]$ covers $\kappa$ elements chosen u.a.r.}
        \label{fig:first}
    \end{subfigure}
    \hfill
    \begin{subfigure}{0.4\textwidth}
        \includegraphics[width=\textwidth]{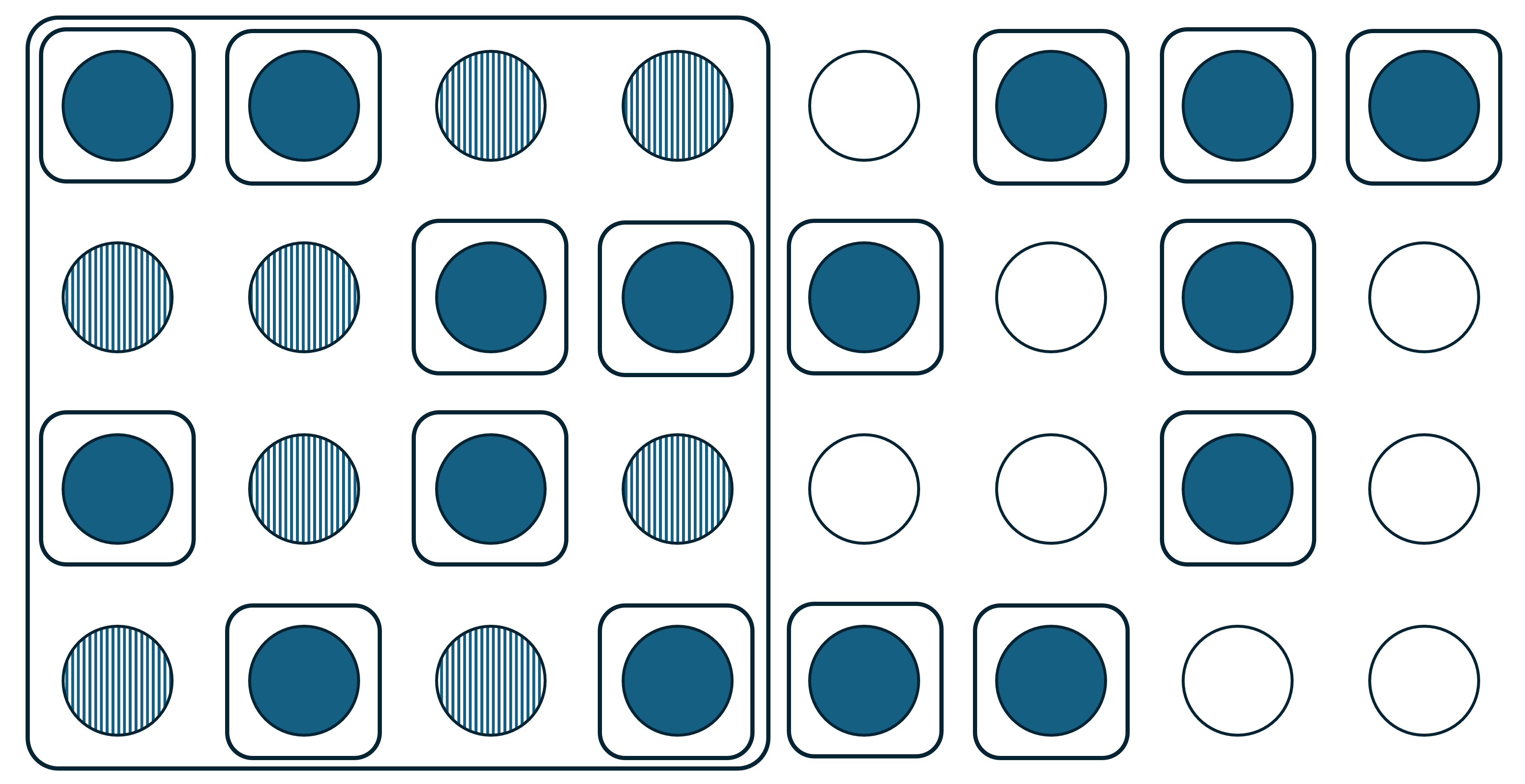}
        \caption{Elements covered by $\cA[V_{\textsf{now}}] \cup R$.}
        \label{fig:second}
    \end{subfigure}
    \caption{Visualization of \ExpoAlg on the perfect alignment instance.}
    \label{fig:info_theoretic}
    \end{figure}

\subsection{Efficient Randomized Algorithm}
\label{subsec:poly-time}

    The LP-based algorithms in the previous section crucially rely on enumerating all possible futures and their corresponding optimal solutions.\footnote{More precisely, \ExpoAlg considers $|\mathcal{R}| \in O\left((\nicefrac{|\Vnow|}{\varepsilon^2})^{2^{|\Vnow|}}\right)$ for general monotone submodular functions. See Section~\ref{sec:minmaxalg} for details.} If we move our attention to polynomial-time algorithms we need a different approach to achieve a better than $\nicefrac 12$-addition robust algorithm (especially since it is $\mathsf{NP}$-hard to match the information theoretic $\nicefrac 23$, in light of the $\nicefrac{(e-1)}{e}$-hardness of approximation in the standard submodular maximization setting \citep{Feige98}). Our algorithm, \GreedyCertificate, draws inspiration from the following observation for coverage functions.
    
    \begin{observation}[Hedging through Overprovisioning]
        As we have seen in \Cref{ex:alignment}, for any fixed solution $\cA[V_{\textsf{now}}]$ that the algorithm picks, the adversary may reveal a future set $R$ such that these two sets cover exactly the same area. In this scenario the algorithm cannot have better than $\nicefrac{1}{2}$-addition robustness as anything it has achieved so far is redundant. To circumvent this {\it perfect alignment} problem, we first compute a larger augmented set $\AugmentedSol$ (see \Cref{fig:augmented}) and make our selection robust by randomly sampling $\kappa$ elements from $\AugmentedSol$ (see \Cref{fig:sampled}). For instance, instead of directly selecting a set of size $\kappa$ with a certain value $v$, we could find $\AugmentedSol$ of size $1.25 \cdot \kappa$ of value $1.2 v$.  Selecting $\kappa$ elements from $\AugmentedSol$ randomly reduces the expected value to $0.96 v$, but addresses the perfect alignment problem, as now the adversary cannot cover \emph{exactly} $\cA[V_{\textsf{now}}]$ with $R$ (see \Cref{fig:final}) and the overall approximation factor can finally overcome the $\nicefrac 12$ barrier.
    \end{observation}
    
    \begin{figure}
\centering
\begin{subfigure}{0.4\textwidth}
    \includegraphics[width=\textwidth]{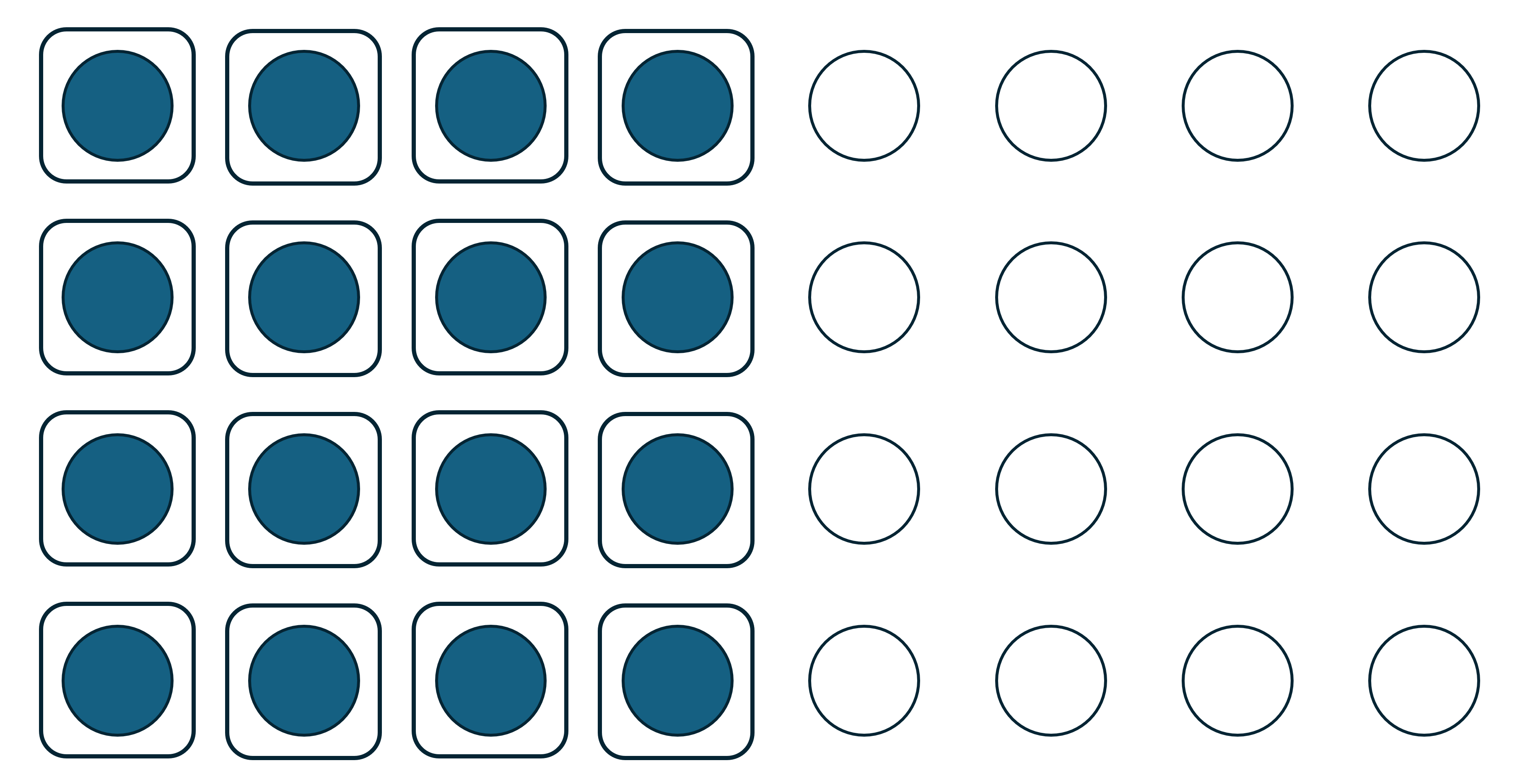}
    \caption{Initial greedy solution $S$.}
    \label{fig:greedy}
\end{subfigure}
\hfill
\begin{subfigure}{0.4\textwidth}
    \includegraphics[width=\textwidth]{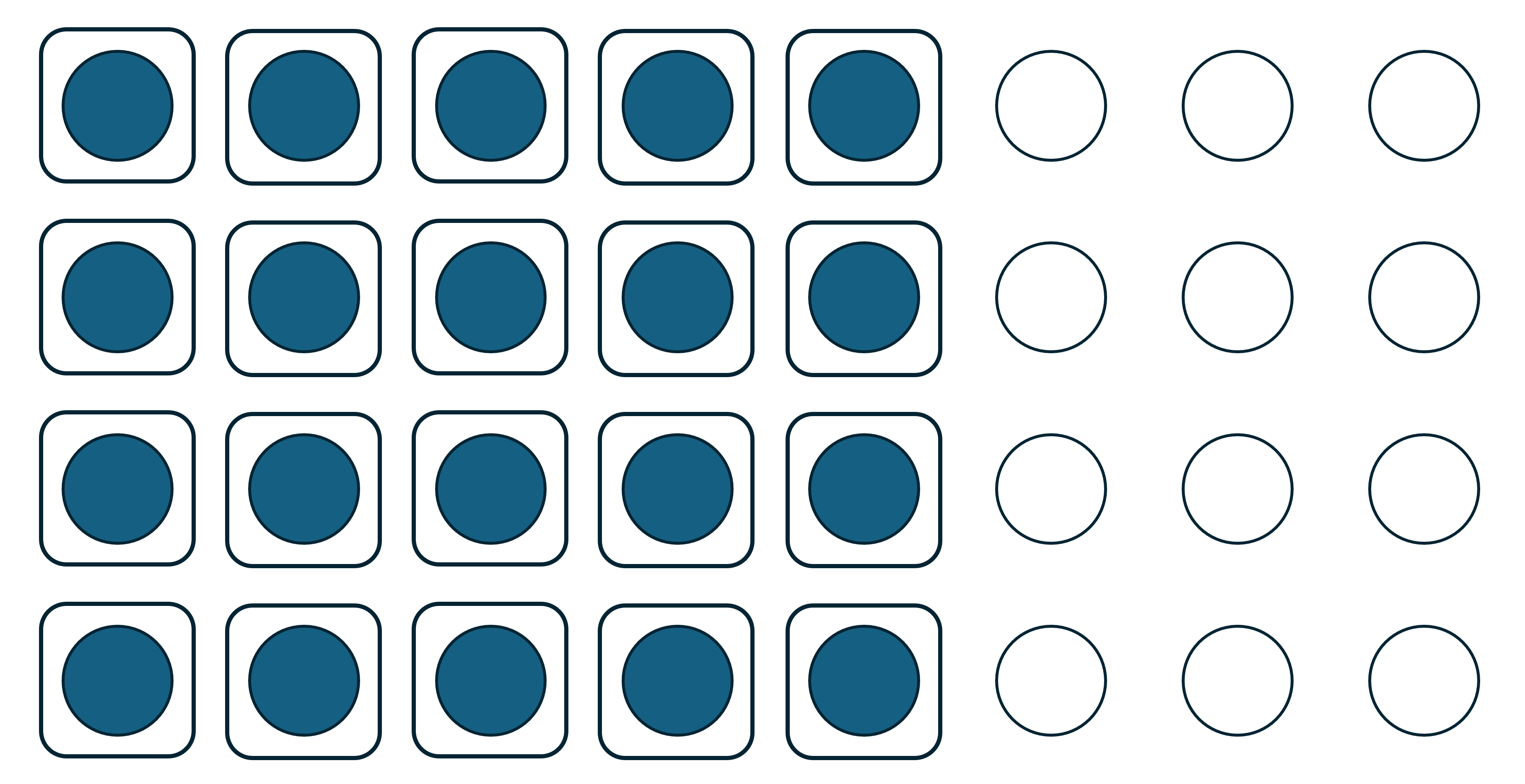}
    \caption{Augmented solution $\AugmentedSol$.}
    \label{fig:augmented}
\end{subfigure}
\begin{subfigure}{0.4\textwidth}
    \includegraphics[width=\textwidth]{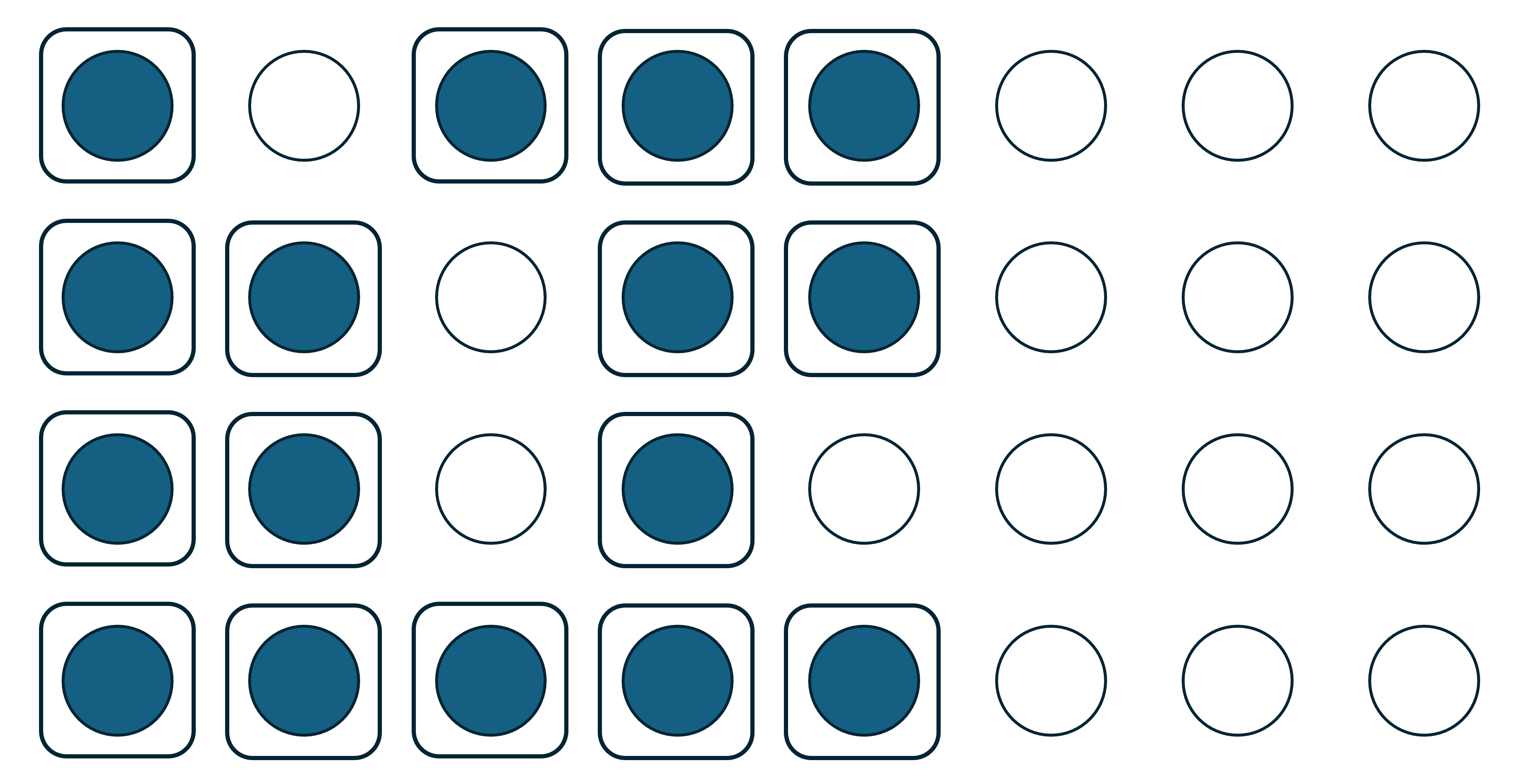}
    \caption{Subsampled solution $A$.}
    \label{fig:sampled}
\end{subfigure}
\hfill
\begin{subfigure}{0.4\textwidth}
    \includegraphics[width=\textwidth]{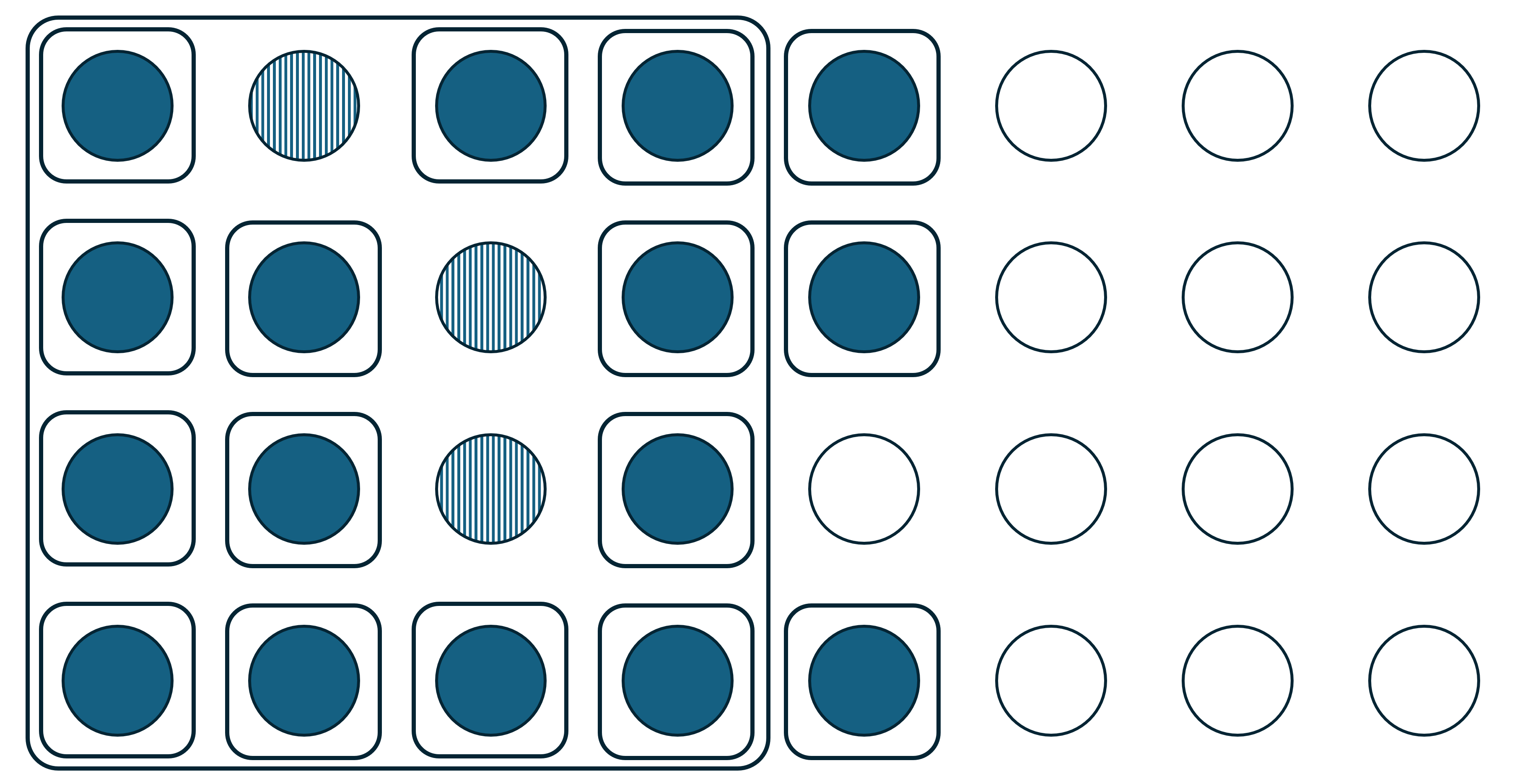}
    \caption{Elements covered by $A \cup R$.}
    \label{fig:final}
\end{subfigure}
        
\caption{Visualization of the \GreedyCertificate algorithm.}
\label{fig:figures}
\end{figure}
    
    More precisely, \GreedyCertificate first computes a greedy solution $S$ of $\kappa$ elements in $V_{\textsf{now}}$ (we refer to \Cref{fig:greedy} for visualization), then augment it to $\AugmentedSol$ (see \Cref{fig:augmented}) by adding at most $\eta \kappa$ extra elements that have high enough marginal contribution to $\AugmentedSol$, i.e., such that 
    \[
        f(e|\AugmentedSol) \ge \frac{\gamma}{\kappa} f(S).
    \]
    Finally, it subsamples a set $A$ of $\kappa$ elements from $\AugmentedSol$ (see \Cref{fig:sampled}). In \Cref{sec:poly-time} we perform a careful optimization of the parameters $\eta$ and $\gamma$ to show a $0.51$ approximation. Here we provide some intuition about the technical challenges that our analysis has to overcome.
    
    Our goal is to approximate the optimal $\kappa$ elements solution $\OPT'$ in $\Vnow \cup R$. The set $\OPT'$ is thus naturally partitioned into its intersections with $\Vnow$ and $R \subseteq \Vfuture$. Given the submodularity of $f$, the most natural approach suggests to analyze separately these two parts:
    \begin{equation}
    \label{eq:naive}
        f(\OPT') \le \underbrace{f(\OPT' \cap \Vnow)}_{\le \nicefrac{(e-1)}{e} \cdot f(S)} + \underbrace{f(\OPT' \cap  \Vfuture)}_{\le f(R)},
    \end{equation}
    where the upper bound on the first term follows from the approximation guarantees of the initial greedy solution, while the second one from observing that $\OPT' \cap  \Vfuture \subseteq R.$ Such a crude analysis, however, cannot achieve a better than $(\nicefrac{(e-1)}{e}+1)^{-1}$-approximation on its own (a priori, $f(\AugmentedSol \cup R)$ may be approximately equal to $f(R)$ and $f(S)$). To improve on this bound, we split the analysis into two cases: either the algorithm fills the extra budget $\eta \kappa$ (\emph{high surviving-value} case), or it is not able to do it, i.e., no element with marginal contribution to $\AugmentedSol$ larger than the threshold survived (\emph{low surviving-value} case).
    
    \paragraph{Low surviving-value case.} In the low surviving-value case, we have \emph{a certificate} that all elements in the optimal solution $\OPT'$ have small marginal contribution with respect to $\AugmentedSol$. Moreover,  only $\eta' \kappa < \eta \kappa$ elements have been added to $S$. As a first observation, it is possible to employ linearity of expectation and submodularity to relate the value of the actual solution $\E{f(A \cup R)}$ with that of the non-subsampled set $f(\AugmentedSol \cup R)$ and of $R$:
    \begin{equation}
        \label{eq:subsampling_intro}
        \tfrac{1}{(1+\eta') }f(\AugmentedSol \cup R)  + \tfrac{\eta'}{(1+\eta') } f(R)\le \E{f(A \cup R)}.
    \end{equation}
    Now, we need to quantify the drop in average value per element going from $S$ to $\AugmentedSol \cup R$. Here, the marginal value of any element added to $\AugmentedSol$ is at least a $\gamma$ fraction of the average value of elements in $S$. Formally, we can prove that 
    \[
        (1+\gamma \eta') f(S) \le f(\AugmentedSol \cup R).  
    \]
    Combining the last two inequalities with Inequality~\ref{eq:naive} (and using that $\eta' < \eta$), we get an approximation factor that is parameterized by the threshold $\gamma$ and the sampling budget $\eta$:
    \begin{equation}
        \label{eq:first_term_intro}
        f(\OPT') \le (1+\ExtraSample)\left(1+  \frac{\threshold}{1+\threshold \ExtraSample}
      \right) \E{f(\FinalSol\cup R)}
    \end{equation}

    \paragraph{High surviving-value case.} The analysis of the high surviving-value case is more challenging, as we do not have any direct way of upper bounding the marginal contribution of the generic $x \in \OPT'$ with respect to $\AugmentedSol \cup S$. The crucial ingredient we need is a way to relate the quality of the initial greedy solution with that of the other elements in $\Vfuture$ that may be added to the optimal solution $\OPT'$. To this end, let $\mu \in [0,1]$ be such that
        \[
            \max_{e \in \Vnow \setminus S} f(e \mid S) = \frac{\mu}{\kappa}f(\OPT_{\textsf{now}}),
        \]
    where $S$ is the greedy solution that constitutes the initial seed of $\AugmentedSol$, and $\OPT_{\textsf{now}}$ is the best $\kappa$-elements set in $\Vnow.$ Intuitively, if $\mu$ is close to $1$, then it means that the actual value of $S$ is larger than the one guaranteed by the worst-case analysis of greedy. Conversely, if $\mu$ is small, then we can exploit that the marginal contribution of all the elements in $\oldOPT$ with respect to $S$ (and thus $\AugmentedSol$ by submodularity) is at most $\nicefrac{\mu}{\kappa}f(\oldOPT)$.
    We can formalize the latter consideration and observe that the marginal contribution of any $x \in \OPT'$ with respect to $S$ (and thus also with respect to $\AugmentedSol \cup R$, by submodularity) is at most $\nicefrac{\mu}{\kappa}f(\oldOPT).$ Therefore, the following inequality holds:
    \[
        \label{eq:improved_bound_intro}
        f(\OPT') \le f(\AugmentedSol \cup R) + \mu f(\oldOPT).
    \]
    We can combine this inequality with Inequality~\ref{eq:naive} and the subsampling Inequality~\ref{eq:subsampling_intro} (in this case $\eta'= \eta$) to get the following approximation bound:
    \begin{equation}
        \label{eq:second_term_intro}
        \frac{f(\OPT')}{\E{f(\FinalSol \cup R)}} \le (1+\ExtraSample) \frac{\min\{f(\oldOPT) + f(R),f(\AugmentedSol \cup R) + \mu f(\oldOPT)\}}{f(\AugmentedSol \cup R) + \ExtraSample f(R)}.
    \end{equation}

    \paragraph{Putting everything together.} As a last step of the analysis, we optimize for the parameters $\gamma$ and $\eta$ that appears in Equations~\ref{eq:first_term_intro} and \ref{eq:second_term_intro}. Note, we need to set these parameters so that the right approximation factor holds for both cases and for all values of $f(\oldOPT), f(\AugmentedSol \cup R), f(R)$ and $\mu$. 
    In finetuning these parameters, we need a final ingredient to achieve the better-than-$\nicefrac12$ approximation factor. In particular, in \Cref{app:refined_greedy}, we propose an improved instance-dependent analysis of the approximation guarantees of the greedy algorithm that may be of independent interest (note, for $\mu= \nicefrac{1}{e}$ we recover the standard $\nicefrac{(e-1)}{e}$-approximation:
    \[
            f(\GreedySol) \ge (1+\mu \ln \mu) f(\oldOPT).
    \]
    This inequality formalizes the intuition that a large $\mu$ correspond to an improved approximation guarantee on $S$, and relates the value of $f(\AugmentedSol \cup R)$ with that of $f(\oldOPT)$ in a non-trivial way. For all missing details we refer to \Cref{sec:poly-time}.
    
    % \textcolor{red}{[Paul: Shall we add a paragraph / subsection ``Organization of Remainder''? ... to say which results appear where]}

%% file: sections/30-meta_algorithm.tex
\section{Reduction to Addition-Robust Submodular Maximization}
\label{sec:addition-robust}

In this section, we present a reduction from the problem of designing a randomized consistent algorithm, to the problem of designing an addition-robust algorithm. Towards this end, we construct a consistent meta algorithm, \checkpoint, and study its properties.

A crucial role is played by the \swap routine, which interpolates between two solutions by exchanging  $\ell$ elements at a time. \swap takes as input two sets, a base set $A$ and a target set $B$, and swaps $\ell$ arbitrary elements in $A \setminus B$ with $\ell$ arbitrary elements from  $B \setminus A$. Then, it returns this new version of $A$, which is closer to the target $B$ (see pseudocode).
    \begin{algorithm}[t]
        \caption*{\swap}
        \begin{algorithmic}[1]
        \STATE \textbf{Input:} Base set $A$ and target set $B$ with the same cardinality, integer $\ell$
        \STATE $A' \gets A \setminus B$
        \STATE $B' \gets B \setminus A$
        \IF{$|A'| \le \ell$}
            \RETURN $B$
        \ELSE
            \STATE Let $A''$ be an arbitrary subset of $A'$ of cardinality $\ell$ \STATE Let $B''$ be an arbitrary subset of $B'$ of cardinality $\ell$
            \RETURN $A \setminus A'' \cup B''$
        \ENDIF
        \end{algorithmic}
        \end{algorithm}
        
        The meta-algorithm \checkpoint divides the stream of insertions into blocks of $\block = \e k$ elements. These blocks are separated by check-points (indexed with $\blockid$), where an $\alpha$-addition robust submodular maximization routine $\cA$ is used to compute a suitable solution of cardinality $\k = (1-2\e)k$ on the elements arrived so far. \checkpoint interpolates between a ``new'' solution $\Snew$ computed at the last check-point $\blockid$ (line~\ref{line:Snew} in the pseudocode), and an ``old'' one $\Sold$ coming from the previous check-point $\blockid'$ (line~\ref{line:Sold}) in the following way. It divides each block (of length $\block = \e k$) into $\nicefrac {1}\e$ equal contiguous sub-blocks and then select one of them in each block uniformly at random (line~\ref{line:sampling}). 
        In each block, the transition happens within the selected sub-block: \swap is called in each of the $\e \block = \e^2 k$ time steps and swaps at most $\nicefrac 1{\e^2}$ elements each time.  
        Besides, \checkpoint maintains in the solution the elements arrived up to the old checkpoint (lines~\ref{line:R_update_1}, \ref{line:R_update_2}, and \ref{line:R_update_3}), to account for drastic changes in the dynamic optimum that may have occurred since the last check-point. We denote the recent elements with $R$ (note, $|R|$ is always at most  $2 \e k$). See the pseudocode for further details. Here, and in the rest of the paper, we make the simplifying assumption that $\nicefrac 1{\e}$ and $\e k$ are integer; this is without loss of generality, as all the arguments can be made formal by considering the integer part. \begin{algorithm}[t]
        \caption*{\checkpoint}
        \begin{algorithmic}[1]
        \STATE \textbf{Environment:} Stream $X$ of $n$ elements, function $f$, cardinality $k$
        \STATE \textbf{Input:} Precision parameter $\e$ and addition-robust submodular routine $\cA$
        \STATE \textbf{Initialization:} $R \gets \emptyset$, $\Sold \gets \emptyset$, $\Snew \gets \emptyset$ $\ALG \gets \emptyset$
        \STATE $\block \gets \e \cdot k$, $\kappa \gets (1-2\e)k$ \hfill \COMMENT{Block length $\Delta$ and addition-robust cardinality $\kappa$}
        \FOR{$t = 1, \dots, \block$}
            \STATE $R \gets R + x_t$, $\ALG \gets R$ \label{line:R_update_1}
        \ENDFOR
        \FOR{$ i = 1, \dots, \nicefrac n {\block}$}
            \STATE $\blockid \gets i \cdot \block$, $\blockid' \gets (i-1) \cdot \block$ \hfill \COMMENT{$\blockid'$ denotes the previous checkpoint}
            \STATE $\Sold \gets \Snew$ \hfill  \COMMENT{The previous block's $\Snew$ becomes the current $\Sold$} \label{line:Sold}
            \STATE Let $\Snew$ be the output of $\cA$ on $X_{\blockid}$ 
            \hfill \COMMENT{Note, $|\Snew| \le \k$}\label{line:Snew}
            \STATE $R \gets R \setminus X_{\blockid'}$, $\ALG \gets \Sold \cup R$ \hfill \COMMENT{$X_{0} \gets \emptyset$ by convention} \label{line:R_update_2}
            \STATE Draw $j$ uniformly at random in $\{0, 1, \dots, \nicefrac 1 \e-1\}$ \label{line:sampling}
            \FOR{$t = \blockid +1, \dots, \blockid + \block $}
            \IF{$t \in \{\blockid + j \e \block +1 , \blockid + (j+1) \e \block \}$}
                \STATE $\Sold \gets \swap(\Sold,\Snew,\nicefrac{1}{\e^2})$ \hfill \COMMENT{At the end of the sub-block, $\Sold=\Snew$} \label{line:swap}
            \ENDIF
                \STATE $R \gets R + x_t$\label{line:R_update_3}
                \STATE $\ALG \gets R \cup \Sold$
            \ENDFOR
        \ENDFOR
        \end{algorithmic}
        \end{algorithm}

        \reduction*
        
        The proof of the theorem is divided into two steps. First,  
        we argue that the randomized choice of the transition sub-blocks guarantees that \checkpoint is consistent.
        
        \begin{lemma}
            \label{lem:check-point-consistency}
            \checkpoint is $O(\nicefrac 1{\e^2})$-consistent
        \end{lemma}
        \begin{proof}
            We can divide the consistency analysis into three cases: the insertions corresponding to check-points, the ones falling in one of the sub-blocks where \swap is called, and all the remaining insertions.
            
            We start a generic checkpoint $\blockid$. At the end of each block of insertions, the current solution $\ALG$ is equal to the last $2 \block$ elements (set of recent elements $R$), plus the new solution in that block. When a new element $x_{\blockid}$ is inserted, corresponding to a check-point, then the previous new solution becomes the old solution, so what happens in line~\ref{line:R_update_2} only entails removing from the solution all the $\block = \e k$ elements in $R \cap X_{\blockid'}$, where $\blockid'$ denotes the previous checkpoint (if any).
            % \footnote{The fact that we are dropping $\Omega(k)$ elements might suggest that our choice of modelling consistency by counting the insertions and not the cardinality of the symmetric difference is with loss of generality.  }

            Consider now what happens upon every insertion that results in a \swap call. The solution is modified in lines~\ref{line:swap} ($\nicefrac 1{\e^2}$ elements are modified) and \ref{line:R_update_3} (where only one recent element is inserted). Overall, we have a $\nicefrac 1{\e^2} +1$ bound on the consistency. 
            Finally, after all the other insertions, the only change in the solution is given by the insertion of the new element in $\ALG$ (line\ref{line:R_update_3}).
        \end{proof}

        It remains to show that using an $\alpha$-addition robust submodular algorithm $\cA$ as a subroutine provides the desired approximation guarantees. Note, $\cA$ is called with cardinality constraint $\k$, so to leave some extra room for the future elements $R$.

        \begin{lemma}
        \label{lem:addition-robust}
            If the submodular routine $\cA$ is $\alpha$-addition-robust, then \checkpoint provides a $(1-2\e)^2\alpha$ approximation to the dynamic optimum.
        \end{lemma}
        \begin{proof}
            Consider any element $x_t$, we want to prove that the solution $\ALG_t$ maintained by the algorithm after the insertion of $x_t$ is a good approximation of the optimal solution of (cardinality $k$) on $X_t$: $\OPT_t$. Insertion $x_t$ belongs to some block starting in checkpoint $\blockid$ and in some sub-block starting in $\blockid + j \e \block.$ Fix the randomness of the algorithm up to checkpoint $\blockid$ (so that $\Snew$ and $\Sold$ are deterministically induced by the past history), and consider the random choice of the sub-block where \swap is used. We have three cases: either \swap is used in the sub-block of $x_t$ (event $\cE_=$), before $x_t$ (event $\cE_<$), or after $x_t$ (event $\cE_>$). We have then:
            \begin{align}
            \nonumber
                    \E{f(\ALG_t)} &= \P{\cE_<} \E{f(\ALG_t)|\cE_<} + \P{\cE_>} \E{f(\ALG_t)|\cE_>} + \P{\cE_=} \E{f(\ALG_t)|\cE_=}\\
                    \label{eq:old_new}
                    &\ge f(\Snew \cup R) \P{\cE_<} + f(\Sold \cup R) \P{\cE_>}.
            \end{align}
            Note, in the above inequality, we denote with $\Sold$ the actual solution computed in the previous checkpoint (as in line~\ref{line:Sold}), not its ``intermediate versions'' as in the iterations of line~\ref{line:swap}.
            The definition of addition-robust algorithm relates the value of the solutions computed in the checkpoints with that of the best solution of cardinality $\k = (1-2\e) k$; we denote with $\OPT'_t$ such solution. By submodularity, linearity of expectation, and a simple averaging argument, it holds that 
            \begin{equation}
                \label{eq:OPT_t'}
                f(\OPT_t') \ge (1-2\e) f(\OPT_t)
            \end{equation}
            Consider now Inequality~\ref{eq:old_new}, we can take the expectation with respect to the rest of the story of the algorithm, and exploit that $\cE_<$ and $\cE_>$ are independent from the past:
            \begin{align}
                    \E{f(\ALG_t)} &\ge \E{f(\Snew \cup R)} \P{\cE_<} + \E{f(\Sold \cup R)} \P{\cE_>} \tag{By Inequality~\ref{eq:old_new}}\\
                    &\ge {\alpha} f(\OPT'_t) (\P{\cE_<} + \P{\cE_>}) \tag{By addition-robustness} \\
                    &\ge{\alpha} {(1-2\e)} f(\OPT_t) (1-\P{\cE_=} ) \tag{By Inequality~\ref{eq:OPT_t'}}\\
                    \nonumber
                    &\ge \alpha(1-2\e)^2 f(\OPT_t),
            \end{align}
            where in the last inequality, we used that the probability of the swap happening in the sub-block containing $x_t$ (i.e., event $\cE_=$) has probability exactly $\e$.
        \end{proof}

%% file: sections/40-exp-time.tex
\section{Tight Information Theoretic Bounds}
\label{sec:expo}
%\ffcomment{Shall we add a sentence here to introduce the Section?}
%\pdcomment{What do you think about making this section about ``Near-Tight Information Theoretic Bounds'' (or similar). The introductory paragraph could emphasize that we are (almost) nailing down what one can achieve ignoring computational efficiency; and then the section could have tow parts (what's currently in this section, plus what's in the next section.)}
%\ola{I like that suggestion} \pdcomment{Great! For now I changed the title, and made the previous sections subsection. You probably want to edit the titles / text.}
In this section, for every $\e>0$, we give a $O(\nicefrac{1}{\e^2})$-consistent, $(\nicefrac{2}{3}-O(\e))$-approximation algorithm and also prove that it is tight up to a factor $(1+O(\e))$.
\subsection{Tight Consistent Algorithm}
\label{sec:minmaxalg}

%\pdcomment{COPIED FROM OLA'S NOTES, PREVIOUSLY APPENDIX F HERE}
%\ola{While this section is complete in terms of proofs, we need to adapt it to improve it, adapt the notation, and in particular use $\alpha$-robust as defined in the previous section }
\iffalse
The setting:
\begin{itemize}
    \item A set $N$ of items have arrived.
    \item We wish to sample $A$ from a distribution $\mu$ over subsets of $N$ of size at most $k$ so that
    \begin{gather*}
        \EO_{A \sim \mu}{f(A \cup R)}  \geq \alpha \cdot \max_{A \subseteq N: |A| = k} f(R \cup A) \qquad \mbox{ for all possible future $R$'s.}
    \end{gather*}
    \item Our goal is to find the distribution $\mu$ that maximizes $\alpha$, which will be our approximation ratio.
\end{itemize}
\fi

We design a $(\nicefrac{2}{3} - \e)$-addition robust algorithm, called \ExpoAlg, and use \Cref{thm:reduction} to achieve a consistent $(\nicefrac{2}{3}-O(\e))$ approximation algorithm.
In other words, we describe an algorithm for the addition-robust problem that samples a set $A$ %\pdcomment{of size at most $k$? of size $k$?} 
from a distribution $\DistributionA$
over subsets of $V_{\textsf{now}}$ of size at most $\kappa$ so that 
\begin{align*}
    \EO_{A \sim \DistributionA }[{f(A \cup R)}]  \geq (\nicefrac{2}{3} -\e)\cdot \max_{\substack{V' \subseteq V_{\textsf{now} }\cup R \\ {|V'| \leq \kappa}}} f(V') \qquad \mbox{ for every $R \subseteq V_{\textsf{future}}$.}
\end{align*}
We remark that the set $V_{\textsf{future}}$ is unknown to the algorithm and in particular we do not know the value of the function $f$ for sets {containing any of those future elements}. To fix this issue, we construct a set $\mathcal{R}$ that intuitively contains an element for 
every possible future set $R$.
To ensure that $\mathcal{R}$ is a finite set, we allow for a small error and discretize the possible values that $f$ takes on these sets.
More formally, we prove the following guarantee:
\begin{align}
    \EO_{A \sim \DistributionA }[{f(A \cup R)}]  \geq \alpha \cdot \max_{\substack{V' \subseteq V_{\textsf{now}} \\ {|V'| \leq \kappa}}} f(V' \cup R) \qquad \mbox{ for every $R \subseteq V_{\textsf{future}}$.}
    \label{eq:strong_robustness_exp_alg}
\end{align}
Note, this is a stronger guarantee than the one requested in the definition of $\alpha$-addition robustness as $f$ is monotone and the set $R$ is now given for ``free'' on the right-hand-side, i.e., it is not counted towards the cardinality constraint $\kappa$.

A technical difficulty is that the algorithm needs to calculate the distribution 
$\DistributionA$  
 and sample $A$ from 
 $\DistributionA$  
 without the knowledge of $V_{\textsf{future}}$ (even without  knowing $|V_{\textsf{future}}|$). 
We solve this difficulty as follows. 
First note that for Inequality~\ref{eq:strong_robustness_exp_alg}, we are only interested in values of $f$ of the form $f(S \cup R)$ for $S\subseteq V_{\textsf{now}}$, i.e., we are never considering a non-trivial subset of $R$. We can therefore think of $R$ as a single element that we denote by $r$. Now, as we do not know $V_{\textsf{future}}$, we extend $f$ to a new function $\hat{f}$ over all possible scenarios of $r$ (or equivalently, $R$) that are consistent with $V_{\textsf{now}}$ (i.e., $\hat{f}(S) = f(S)$ for $S \subseteq V_{\textsf{now}}$ and $\hat{f}$ is a monotone submodular function on $V_{\textsf{now}} \cup \{r\})$.  The number of  potential scenarios is infinite but, by a standard discretization argument, we can make the number of possible scenarios to be finite at the cost of an additional $O(\e)$ term in the approximation ratio.  This discretized set of scenarios is the set $\mathcal{R}$ in the following lemma, whose formal proof is deferred to \Cref{sec:app-expo}.

\begin{restatable}{lemma}{lemdisc}
For every $\e >0$, we can compute a finite set $\mathcal{R}$   and an extension $\hat{f}$ of $f$ to the domain $2^{V_{\textsf{now}} \cup \mathcal{R}}$ with the  following guarantees:
\begin{enumerate}
    \item For every $r\in \mathcal{R}$, $\hat{f}$ is non-negative, monotone, and submodular when restricted to $V_{\textsf{now}} \cup \{r\}$.
    \item For $0\leq \alpha \leq 1$, if
    \begin{align*}
    \EO_{A \sim \DistributionA }[{\hat{f}(A \cup \{r\})}]  \geq \alpha \cdot \max_{\substack{V' \subseteq V_{\textsf{now}} \\ {|V'| \leq \kappa}}} \hat{f}(V' \cup \{r\}) \qquad \mbox{ for every $r \in \mathcal{R}$.}
    \end{align*}
    then
    \begin{align*}
    \EO_{A \sim \DistributionA }[{f(A \cup R)}]  \geq (\alpha- \e) \cdot \max_{\substack{V' \subseteq V_{\textsf{now}}\\ 
    p{|V'| \leq \kappa}}} f(V' \cup R) \qquad \mbox{ for every $R \subseteq V_{\textsf{future}}$.}
    \end{align*}
\end{enumerate}
%Moreover, if $f(\cdot)$ is a coverage function then we can choose $\mathcal{R}$ to 
Moreover, in the special case when $f$ is a coverage function,  we can choose the extension $\hat{f}$ and the finite set $\mathcal{R}$ so that $\hat f$ is a  coverage function on the domain $2^{V_{\textsf{now}} \cup \mathcal{R}}$.

\label{lemma:discretization}
\end{restatable}
We remark that $|\mathcal{R}| \in O\left((\nicefrac{|V_{\textsf{now}}|}{\varepsilon^2})^{2^{|\Vnow|}}\right)$ for general functions and 
$|\mathcal{R}| \in O(2^{|Y|})$
in the case of a coverage function with underlying universe $Y$; however the exact bound is irrelevant for our information theoretic arguments and we only use that $|\mathcal{R}|$ is finite.
%By allowing for a small error, we 
%Two point about the above inequality needs to be clarified.
%First, the above inequality is with we need to define all the possible future set $R$, which is all the subsets of the to the algorithm unknown set $V_{\textsf{future}}$.
%Second is that we do not know the value of function $f$ for all these values. To fix this issue, we construct an infinite size set of all the possible submodular functions for each of the possible future set $R$ and one of them will be the actual set $R$ and submodular function that we will face in the future. Therefore we get the guarantee that our algorithm performs well for that realization of the future as well. To be more clear, we create infinitely many pairs of $R$ and submodular functions $f$ and we find a distribution $\mathcal{D}$ that provides a good solution for all of them. 
Equipped with the above lemma, our goal is now to find a distribution $\DistributionA$ that maximizes $\alpha$ {uniformly} in {the following inequalities}
    \begin{align*}
    \EO_{A \sim \DistributionA }[{\hat{f}(A \cup \{r\})}]  \geq \alpha \cdot \max_{\substack{V' \subseteq V_{\textsf{now}} \\ {|V'|\leq \kappa}}} \hat{f}(V' \cup \{r\}), \qquad \mbox{ $\forall$ $r \in \mathcal{R}$.}
    \end{align*}
    We do so by solving a linear program and show that its solution is bounded by $\nicefrac 23$, i.e., there exists a feasible solution such that $\alpha \geq \nicefrac 23$ (while $\alpha \geq \nicefrac 34$ if $\hat f$ is a coverage function).  To that end, for any $r\in \mathcal{R}$, let
\[
\OPT(r) = \max_{\substack{V' \subseteq V_{\textsf{now}} \\ {|V'| \leq \kappa}}} \hat f(V' \cup \{r\})\,,
\]
%
%
%\morteza{we have A on both left and right sides. On the left side, it is drawn from a distribution, and on the right it is the maximizer variable per set R. I find it a little bit confusing. Do you think we should use separate symbols? If so, any suggestions? We have also used $X_\tau$ in other sections. Do we want to keep $X$ as it is. We can keep it as it is but we should be aware of the slight discrepancy.}
and write the following linear program:
\begin{align*}
    \text{Maximize} & \quad \alpha \\[0.2cm]
    \sum_{A \subseteq V_{\textsf{now}}: |A| \leq \kappa} \lambda_A \cdot \hat f(A \cup \{r\}) & \geq \alpha \cdot  \OPT( r) \qquad \mbox{for $r\in \mathcal{R}$}\\
    \sum_{A \subseteq V_{\textsf{now}}: |A| \leq \kappa} \lambda_A & \leq 1 \\
    \lambda_A & \geq 0 \qquad \mbox{for all $A \subseteq V_\textsf{now}: |A| \leq \kappa$.}
\end{align*}
Algorithm \ExpoAlg, first computes $\mathcal{R}$ via \Cref{lemma:discretization}, it then solves the above linear program and defines $\DistributionA$ as the distribution that samples set $A \subseteq V_{\textsf{now}}$ with probability $\lambda_A$ (and the emptyset with probability $1- \sum_A \lambda_A$).  

We continue to bound the guarantee $\alpha$ of \ExpoAlg by considering the dual linear program. If we associate a variable $y_r$ with the constraint associated to $r$ and a variable $z$ for the constraint $\sum_{A} \lambda_A  \leq 1$,
we get the dual
\begin{align*}
    \text{Minimize} & \quad z \\[0.2cm]
    \sum_{r\in \mathcal{R}} y_r\cdot  \OPT(r) & = 1\\
     \sum_{r\in \mathcal{R}} y_r\cdot  \hat f(A \cup \{r\}) & \leq  z\qquad \mbox{ for all $A \subseteq V_{\textsf{now}}, |A| \leq \kappa$} \\
    y,z & \geq 0
\end{align*}

Observe that we do not change the value of the dual if we multiply $y$ by a factor $\zeta$ and $f$ by a factor $\nicefrac{1}\zeta$. We may thus assume that an optimal solution to the dual is such that
\begin{gather*}
     \sum_{r\in \mathcal{R}} y_r\cdot  \OPT(R) = 1 \qquad \mbox{ and } \qquad \sum_{r\in \mathcal{R}} y_r = 1\,.
\end{gather*}
In other words, we are given a distribution over $r\in \mathcal{R}$ and we wish to show that there exists a candidate set $A$ such that $\sum_{r\in \mathcal{R}} y_r \cdot  \hat f(A \cup \{r\})  \geq \alpha$ with $\alpha = \nicefrac 23$. This then completes the proof as this shows that the dual has value $z\geq \alpha$ and thus the guarantee of \ExpoAlg, which uses  the distribution  $\DistributionA$ defined by an optimal solution to the primal linear program, is at least $\alpha$.

For notational convenience, let us in the subsequent assume  that the support of this distribution defined by $y$ is $r_1, \ldots, r_m$ and $y_{r_i} = \nicefrac 1m$.\footnote{This is without loss of generality. In fact, note that $r_1, ..., r_M$ are all the sets with positive $y$ values and let $\e'$ denote their largest common divisor. We can copy each $r_i$, $\nicefrac{y_{r_i}}{\e'}$ times and set their $y$ value to $\e'$ to get the assumption.}
Furthermore, let $A_i \subseteq  V_{\textsf{now}}$ be a set of size at most $\kappa$ so that $\hat f(A_i \cup \{r_i\})  = \OPT(r_i)$. With this notation
\begin{gather*}
     \sum_{r\in \mathcal{R}} y_r\cdot  \OPT(r) = \EO_{i \sim [m]} [\hat f(A_i \cup \{r_i\})]\,.
\end{gather*}
We propose selecting candidate set $A$ for this distribution of $r\in \mathcal{R}$ sets by setting $A = A_j$ where $A_j$ is sampled uniformly at random from $A_1, \ldots, A_m$. So we need to prove that
\begin{gather}
      \EO_{i,j \sim [m]} [\hat f(A_i \cup \{r_j\})] \geq \alpha \cdot \EO_{i \sim [m]} [\hat f(A_i \cup \{r_i\})]
      \label{eq:exp_approx_guarantee}
\end{gather}

\begin{lemma} 
\label{lem:expo-app-submodular}
    Inequality~\ref{eq:exp_approx_guarantee} is satisfied with $\alpha = \nicefrac 23$.
\end{lemma}
\begin{proof}

%    Suppose that $\EO_{i\in [n]} [f(A_i)] \geq \EO_{i\in [n]}[f(A_i  \cup R_i)]/2$. The other case when $\EO_{i\in [n]} [f(R_i)] \geq \EO_{i\in [n]}[f(A_i  \cup R_i)]/2$ is symmetric and omitted. Also note that at least one of these two cases must hold because $f(A_i) + f(R_i) \geq f(A_i  \cup R_i)$ by submodularity.
    %
    
%    \morteza{I don't see where we are using the above argument. Why do we need $\EO_{i\in [n]} [f(A_i)] \geq \EO_{i\in [n]}[f(A_i  \cup R_i)]/2$? The rest of the proof seems to be sufficient and self-contained. Ashkan: I agree.}
    
   As $\hat f$ is submodular when restricted to $V_{\textsf{now}} \cup \{r_j\}$, 
    \begin{align}
    \nonumber
        \EO_{i,j \sim [m]}[ \hat{f}(A_i) + \hat{f}(r_j | A_i) + \hat{f}(A_j|A_i)] &\geq \EO_{i,j\sim [m]}[ \hat{f}(A_i) + \hat{f}(r_j | A_i) + \hat{f}(A_j|A_i \cup \{r_j\})] \\
    \nonumber
        &= \EO_{i,j\sim [m]}[ \hat{f}(A_i \cup \{r_j\} \cup A_j)]\\
    \label{ineq:corr_gap}
        &\ge \EO_{j\sim [m]}[ \hat{f}(A_j  \cup \{r_j\})]  = \OPT,
    \end{align}
    where the last inequality follows by monotonicity. We now split the analysis into two cases, depending on the values of $\EO_{i,j\sim [m]}[\hat{f}(r_j |A_i)]$, $\OPT$, and $\EO_{i\in [m]}[ \hat{f}(A_i)]$.
    
    \paragraph{Case I.} In the first case, we have that $\EO_{i,j\sim [m]}[\hat{f}(r_j |A_i)] \geq \nicefrac{2}{3} \cdot \OPT - \EO_{i\in [m]}[ \hat{f}(A_i)]$. The statement of the Lemma is then immediate:    \begin{gather*}
        \EO_{i,j\sim [m]}[\hat{f}(A_i) + \hat{f}(r_j |A_i)] \geq \frac{2}{3} \OPT.
    \end{gather*}
    
    \paragraph{Case II.} In the second case, we have the converse inequality, which implies the following (by plugging it into Inequality~\ref{ineq:corr_gap}):
    \begin{align*}
         \EO_{i,j\sim [m]}[\hat{f}(A_j|A_i)] \geq \OPT - \EO_{i,j\sim [m]}[\hat{f}(A_i) + \hat{f}(r_j|A_i)] \geq \frac{\OPT}3\,.
    \end{align*}
    We now conclude this case by a symmetrization argument:
    \begin{align*}
        \EO_{i,j\sim [m]}[ \hat{f}(A_i \cup \{r_j\})]  & = \EO_{i,j\sim [m]}[ \hat{f}(r_j) +  \hat{f}(A_i\mid r_j)] \\
        & = \EO_{i,j,k\sim [m]} \left[ \hat{f}(r_j) + \tfrac{\hat{f}(A_i\mid r_j) + \hat{f}(A_k\mid r_j)}{2}\right] \\
        &\geq\EO_{i,j,k\sim [m]} \left[ \hat{f}(r_j) + \tfrac{\hat{f}(A_i \cup A_k\mid r_j)}{2}\right]   \qquad \tag{By submodularity on $V_{\textsf{now}} \cup \{r_j\}$} \\
        & = \EO_{i,j,k\sim [m]} \left[ \tfrac{\hat{f}(r_j) + \hat{f}(A_i \cup A_k)}{2}\right] \\
        & = \EO_{i,j,k\sim [m]} \left[ \tfrac{\hat{f}(r_j) + \hat{f}(A_i)+ \hat{f}(A_k| A_i)}{2}\right] \ge \frac{2}{3}\OPT\,.
    \end{align*}
    For the last inequality we used  that $\EO_{k,i\sim [m]}[\hat{f}(A_k| A_i)] \geq \nicefrac{\OPT}3$ and the following inequality:
    \begin{align*}
        \EO_{i,j\sim [m]}[ \hat{f}(r_j) + \hat{f}(A_i)] &= 
\EO_{j\sim [m]}[ \hat{f}(r_j)]  + \EO_{i\sim [m]}[\hat{f}(A_i)] \\
&= 
\EO_{i\sim [m]}[ \hat{f}(r_i) + \hat{f}(A_i)] \geq
\EO_{i\sim [m]}[ \hat{f}(A_i \cup r_i)]
= \OPT\,.
    \end{align*}    
\end{proof}

By combining \Cref{lem:expo-app-submodular} and \Cref{lemma:discretization} we conclude that \ExpoAlg is a $(\nicefrac{2}{3} - \e)$-addition robust algorithm and so we have the following result via \Cref{thm:reduction}.
\begin{theorem}
\label{thm:exp-time}
    For any $\e \in (0,1)$, there exists a randomized  algorithm that is $O(\nicefrac{1}{\e^2})$ consistent and provides a $(\nicefrac{2}{3}-  O(\e))$-approximation of the dynamic optimum.
\end{theorem}

Let us now focus on the special case when $f$ is a coverage function.  We have the following stronger bound on Inequality~\ref{eq:exp_approx_guarantee} when $f$ is a coverage function.
\begin{lemma}
\label{lem:expo-app-coverage}
    Inequality~\ref{eq:exp_approx_guarantee} is satisfied with $\alpha = \nicefrac 34$ if $f$ is a coverage function.
\end{lemma}
\begin{proof}
    We have by \Cref{lemma:discretization} that $\hat f$ is then a coverage function. Let $Y$ be the underlying universe so that every element in the domain $2^{V_{\textsf{now}} \cup \mathcal{R}}$ of $\hat f$ corresponds to an element of $Y$.  
    
    Consider an element $y$ in the underlying universe $Y$ of the coverage function. Let $p_A$ and $p_R$ be the probability that $y$ is covered when selecting a random solution $A_i$ and a random solution $R_j$, respectively. Then the probability that $y$ is covered by a random solution $A_i, R_i$ is at most $\min\{1, p_A + p_R\}$. The proof is now completed by observing that  the probability that a random solution $A_i, R_j$ covers $y$ equals $p_A  + p_R - p_A\cdot p_R \geq \nicefrac 34 \min\{p_A + p_R, 1\}$. 
\end{proof}

Similarly, by combining \Cref{lem:expo-app-coverage} and \Cref{lemma:discretization} we conclude that \ExpoAlg is a $(\nicefrac{3}{4} - \e)$-addition robust algorithm and so we have the following result via \Cref{thm:reduction}.
\begin{theorem}
\label{thm:exp-time-coverage}
    For any $\e \in (0,1)$ and any coverage function, there exists a randomized algorithm that is $O(\nicefrac{1}{\e^2})$ consistent and provides a $(\nicefrac{3}{4}-  O(\e))$-approximation of the dynamic optimum.
\end{theorem}

\subsection{Information-Theoretic Hardness for Submodular Functions}\label{sec:hardness-submodular}

%\pdcomment{I have moved this section here.}

In this section, we provide a matching
%\pdcomment{Technically it's an upper bound if we go with the convention that $\alpha \leq 1$, let's rewrite and use hardness/impossibility instead. Also in the title of this section.} upper
impossibility bound for the problem. We start by defining a simple  collection of $nm$ submodular functions $g_1, g_2, \ldots, g_m$ that we then build upon to obtain the hardness result. These functions will be defined over the set of elements   $X_g = \{a_1, a_2, \ldots, a_m,r\}$ and they agree on the values of the sets $S \subseteq \{a_1, \ldots, a_m\}$. Intuitively, the $r$ element is yet to arrive in the stream and the consistent algorithm will need to select the elements $\{a_1, \ldots, a_m\}$ without knowing their interactions with element $r$ which is different for each $g_i$. As all $g_i$'s are identical on $\{a_1, \ldots, a_m\}$ the algorithm does not know the index $i$ before the arrival of $r$. We now define $g_i$ formally for $1\leq i \leq m$:  
\begin{itemize}
    \item $g_i(\emptyset) = 0$.
    \item $g_i(a_j) = f(r) = 1$ for $1 \leq j \leq m$.
    \item $g_i(\{a_j, r\}) = \nicefrac  43$ for $1 \leq j \leq m$ and $i \neq j$.
    \item $g_i(\{a_j, a_k\}) = \nicefrac 53$ for $1 \leq j,k \leq m$ and $j \neq k$. %\ola{I updated this from 1.8 to 1.6. I think it was wrong previously.}
    %\silvio{Can an algorithm always play $\{a_i, a_j\}$ and get 0.8?}
    \item  $g_i(\{a_j , a_k , r\})  = \nicefrac 53$, for $1 \leq j,k \leq m$ and $i$, $j$ and $k$ distinct. %\ola{Also updated this case}
%    \item $f(\{a_i , r_i , a_j\} = f(\{a_i , r_i , r_j\} = 2$, for $1 \leq i,j \leq n$ and $i \neq j$.
    \item $g_i(S) = 2$, for any other $S \subseteq X_g$, i.e., if $S$ satisfies one of the following: $|S\cap \{a_1, \ldots, a_m\}| \geq 3$ or $\{a_i, r\} \subseteq S$.
\end{itemize}
One can verify by basic calculations that $g_i$ is a monotone submodular function and that all $g_i$'s are identical on subsets of $\{a_1, \ldots, a_m\}$ (for completeness, we refer to \Cref{sec:app-expo} for a formal proof).

Before continuing with the formal proof where the algorithm can select $k$ elements and swap some of them, for intuition, consider the  case when the algorithm irrevocably selects \emph{only} one of the elements in $\{a_1, \ldots, a_m\}$ before the element $r$ is revealed. That is, first the $\{a_1, \ldots, a_m\}$ elements arrive in an arbitrary order on the stream and the algorithm picks one of them. Afterwards, $r$ is the next element on the stream, and the adversary samples $g_i$ uniformly at random from $\{g_1, \ldots, g_m\}$.  Note that before the arrival of $r$, all $g_1, \ldots, g_m$ are identical and the algorithm does not know which one it is.
Say the algorithm picks an element $a_j$. We then have 
\begin{gather*}
    \EO_{i\sim[m]} [g_i(\{a_j, r\})] = 2 \cdot \frac{1}{m} +  \frac 43 \cdot \frac{ (m-1)}m \leq \frac 43 + \frac 1m\,.
\end{gather*}
At the same time, an optimal solution in hindsight has value $\EO_{i\sim[m]} [g_i(\{a_i, r)] = 2$. It is this gap of $\nicefrac{(\nicefrac 43)}2 = \nicefrac 23$ that is the source of our hardness.

However, the collection $g_1, \ldots, g_m$ does not directly give a hardness result as the algorithm can pick up to $k$ elements and may also swap some of them. Indeed in the above case, it would be sufficient to swap a single element after the arrival of $r$ to obtain the optimal solution.

 To resolve this issue we describe another submodular function $f_i$ for each $g_i$ where each of the elements $a_i$ is divided into multiple smaller elements, a  ``lifting''   technique that was previously used in e.g.~\cite{Vondrak13}. More precisely, we let the set $A_i = \{a_i^1, a_i^2, \ldots, a_i^k\}$ for $1 \leq i \leq m$ and 
\[
X_f = A_1 \cup A_2 \cup \cdots \cup A_m \cup \{r\}.
\]
For each $i = 1, \ldots, m$, we define a submodular function $f_i$ over the elements $X_f$. To that end, we let the function $G_i$ denote the multi-linear extension of the above defined function $g_i$. Before proceeding to defining $g_i$ let us recall the definition of multi-linear extension. For any $m+1$-dimensional vector $y$, the value of $G_i(y)$ is defined as
\[
G_i(y) = \sum_{S \subseteq X_f} \prod_{a_i \in S} y_i \cdot \prod_{r \in S} y_{n+1} \cdot \prod_{a_i \notin S} (1-y_i) \cdot \prod_{r \notin S} (1-y_{m+1}) \cdot f_i(S).
\]
In other words, $G_i(y)$ equals the expected value of $g_i(S)$ where $S$ is obtained by randomly and independently including each $a_i$ with probability $y_i$ and $r$ with probability $y_{m+1}$.
We now define the monotone submodular function $f_i$ for any set $S \subseteq X_f$:
\[
f_i(S) = G_i\left(\frac{|S \cap A_1|}{k}, \frac{|S \cap A_2|}{k}, \cdots , \frac{|S \cap A_n|}{k}, |S \cap r|\right).
\]
It is not hard to verify that $f_i$ is indeed nonnegative, monotone, and submodular using that $g_i$ satisfies these properties (see e.g. Lemma 4.5 in~\cite{FeldmanNSZ23}).
% \ola{In the definition of $g(S)$ we shouldn't divide by $k$ for the $r_i$ eleemnts right? yes, done.}

We now proceed to describe the distribution of hard instances. 
The order of the elements of the stream is similar to before, first elements in $A_1 \cup A_2 \cup \cdots \cup A_m$ appear in arbitrarily order. Then the element $r$ appears and the adversary selects a random function $f_i$ uniformly from $\{f_1, \ldots, f_m\}$. Indeed, note that $f_i$'s are identical on the subsets of $\{a_1, \ldots, a_m\}$ (since $g_i$'s are identical) and so the choice of the specific $f_i$ (i.e. the interaction of $r$ with the other elements) is only necessary at the arrival of $r$. 

\begin{theorem} \label{thm:hardness07}
For any $\e \in (0,1)$, there exists no $\e k$-consistent $(\nicefrac 23+2\e)$-approximation algorithm.  
\end{theorem}
\begin{proof}
Assume toward contradiction that there exists such an algorithm $\mathcal{A}$. 
Let $S^*_e$ denote the solution of $\mathcal{A}$ before the last element $r$ of the stream arrives. We remark that we may assume $S^*_e$ is a deterministic set under the knowledge that $f_i$ is selected uniformly at random from $\{f_1, \ldots, f_n\}$. (Alternatively, the arguments below also applies when we have a distribution over sets $S^*_e$.) After the last element $r$  is revealed, $\mathcal{A}$ updates $S^*_e$ by swapping at most $\e k$ elements. To simplify the proof, we add $r$ and up to $\e k$ elements from $A_{i}$ to the set $S^*_e$ where $i$ is the index of the chosen function $f_i$. We denote the resulting set by $S^*$. We bound the expected value of $f_i(S^*)$ which is an upper bound on the solution chosen by $\mathcal{A}$ by the monotonicity of $f_i$. 
If we let $y_j  = \tfrac{|S^*\cap A_j|}k$ for $i = 1, \ldots, m$,  
\begin{gather*}
    f_i(S^*)  = \sum_{S \subseteq X_f}  \left(\prod_{a_j \in S} y_j \cdot \prod_{a_j \notin S} (1-y_j)\right) \cdot f_i(S \cup \{r\})
\end{gather*}
which by submodularity is upper bounded by
\begin{gather*}
    \sum_{S \subseteq X_f}  \left( \prod_{a_j \in S} y_j \cdot \prod_{a_j \notin S} (1-y_j) \right) \sum_{a_j\in S}f_i(\{a_j, r\}) = \sum_{j=1}^m y_j \cdot f_i(\{a_j, r\})\,.
\end{gather*}
The last equality holds because $\sum_{S \subseteq X_f \setminus j} \prod_{a_k \in S} y_k \cdot \prod_{a_k \not \in S} (1-y_k) = 1$ and thus $f_i(\{a_j, r\})$ appear above with coefficient $y_j$.
%\ashkan{I do not see why they are equal.}
By definition,  $f_i(\{a_j, r\}) = 2$ if $i=j$ and $f_i(\{a_j, r\}) = \nicefrac 43$ otherwise. Hence, over the randomness over the randomly chosen $i\in \{1, \ldots, m\}$, we get
\begin{gather*}
    \EO_{i\sim[m]} [g_i(S^*)] \leq  \sum_{j=1}^m \EO_{i\sim[m]} \left[y_j\cdot f_i(\{a_j, r\})\right] \leq \frac 43 + 2 \cdot \left({\frac 1m + \e }\right)\,,
\end{gather*}
where we used that $\sum_{j=1}^m \tfrac{|S^*_e \cap A_j|}{k}\leq 1$ and that $S^*$ was obtained from $S^*_e$ by adding at most $\e k$ elements from $A_{i}$. Thus $\EO_{i\sim[m]}[ y_{i}] \leq \EO_{i\sim[m]} \left[\frac{|S^*_e \cap A_{i}| + \e k}{k}\right]\leq \nicefrac 1m+\e$.

Finally, let us now bound the value of the optimum solution denoted by \OPT. In the optimum solution we pick $k-1$ elements of $A_{i}$ and $r$. Therefore
\begin{align*}
    \EO_{i\sim[m]}[g_i(\text{\OPT})] \geq \frac{k-1}{k} \EO_{i\sim[m]}[f(\{a_{i} , r\})] \geq \frac{k-1}{k} \cdot 2\,.
\end{align*}
The proof now follows by choosing large enough values of the parameter $m$ and cardinality constraint $k$.
\end{proof} 

\subsection{Information Theoretic Hardness for Coverage Functions}\label{sec:hardness-coverage}

    For coverage function, it is fairly easy to exhibit an hard instance that shows the tightness of the $\nicefrac 34$ factor, using the perfect alignment phenomenon presented in \Cref{ex:alignment}.
    \begin{theorem}
    \label{thm:3/4_randomized}
        For any $\e \in (0,1)$, there exists no $\e k$-consistent $(\nicefrac 34+2\e)$-approximation algorithm, even for coverage functions.  
    \end{theorem}
    
        \begin{proof}
            Fix any precision parameter $\e > 0$, and a (possibly randomized) algorithm $\cA$ that is $\e k$-consistent; we construct a covering instance such that $\cA$ does not maintain a $(\nicefrac 34+\e)$ approximation (in expectation).
            Let $U = \{u_1, \dots, u_{N}\}$ be an universe of elements, and $V$ be the family of subsets of $U$ that contains all $S \subseteq U$ such that $|S| \in \{1,k\}$. The covering function $f$ is naturally defined on $V$, and we consider the task of maximizing $f$ with cardinality $k$, with $N = 2k$.
            
            Observe the behaviour of $\cA$ on the sequence $\{u_1\}, \dots \{u_N\}$. At the end of this partial sequence $\cA$ maintains a certain solution $S = \{\{u_{i_1}\}, \dots, \{u_{i_{\ell}}\}\}$ that, without loss of generality has cardinality $\ell = k.$ The solution $S$ is random, as it may depend on the internal randomization of the algorithm $\cA$; by a simple averaging argument, it holds that there exists $R = \{u_{j_1}, \dots, u_{j_{k}}\}$ such that 
            \begin{equation}
                \label{eq:intersection}
                \E{|R \setminus S|} = \frac {k^2}N = \frac{k}{2}.
            \end{equation}
            
            Now suppose the next element to arrive is $R$. The value of the optimal solution after this insertion is $2k-1$ (just take the last subset and $k-1$ non overlapping singletons). The value of $S$ is $k$ and, even if $\cA$ adds to $S$ the subset $K$ and $\e k$ other singletons, it cannot get a solution of value more than $\nicefrac{3k}2 + \e k$ (by Equation~\ref{eq:intersection}). 
            Overall, if we call $\ALG$ the solution maintained by the algorithm, and $\OPT$ the optimal solution, we have:
            \[
                \frac{\E{f(\ALG)}}{f(\OPT)} \le \frac{3k + 2\e k}{2(2k -1)} 
                \le \frac{3}{4} + 2\e,
            \]
            where the last inequality follows by taking $k$ large enough.
        \end{proof}

%% file: sections/50-poly-time.tex
\section{Breaking the \texorpdfstring{$\nicefrac 12$}\, Barrier with a Poly-Time Algorithm}\label{sec:poly-time}

    We present \GreedyCertificate, which is a $0.51$-addition-robust randomized algorithm that runs in polynomial time. This result, combined with \Cref{thm:reduction}, yields a consistent $(0.51-O(\e))$-approximation algorithm that runs in poly-time. \GreedyCertificate takes two parameters as input: sub-sampling parameter $\ExtraSample$, and threshold $\threshold$ and works as follows. First, it computes a greedy solution $\GreedySol$ of cardinality $\k$, then it augments it to $\AugmentedSol$ by adding up to $\ExtraSample \k$ extra elements, with the property that its marginal contribution is above a certain threshold (depending on $\threshold$). Finally, a random size $\k$ subset $\FinalSol$ of $\AugmentedSol$ is output. We refer to the pseudo-code for further details. Note, the augmentation procedure motivates the name of the algorithm: either the extra sampling budget $\ExtraSample\k$ is filled, or the threshold condition is not met, and thus we have a ``certificate'' on the low marginal value of the remaining elements (as in \Cref{lem:delta>delta'}).   
        \begin{algorithm}[H]
        \caption*{\GreedyCertificate}
        \begin{algorithmic}[1]
        \STATE \textbf{Environment:} Set $X$ of $n$ elements, function $f$, cardinality $\k$
        \STATE \textbf{Input:} Subsampling limit $\ExtraSample$, and threshold $\threshold$. 
        % \STATE \textbf{Initialization:} $\GreedySol \gets \emptyset$
        \STATE Let $\GreedySol$ be the output of $\greedy$ with cardinality constraint $\k$ on set $X$.
        \STATE Initialize $\AugmentedSol$ to $\GreedySol$
        \FOR{$t = 1, \dots, \ExtraSample \cdot \k$}
            \IF{there exists $x \in X$ such that $f(x|\AugmentedSol) \ge \threshold \frac{f(\GreedySol)}{\k}$}\label{line:gamma_condition}
                \STATE Add $x$ to $\AugmentedSol$
            \ELSE
                \STATE Break from the for loop.
            \ENDIF
        \ENDFOR
        \STATE \textbf{Return} a random subset $\FinalSol$ of $\AugmentedSol$ of cardinality $\k$
        \end{algorithmic}
        \end{algorithm}
    
    A crucial ingredient in our analysis of \GreedyCertificate is a refined analysis of the \greedy algorithm for monotone submodular maximization subject to a cardinality constraint \citep{NemhauserWF78}. \greedy builds a solution one step at a time, repeatedly adding the element with largest marginal value with respect to the current solution, and it is known to provide a tight approximation guarantee of $1-\nicefrac 1e$. In the following lemma, we relate the output of \greedy with the value of the optimal solution $\OPT$ and the largest marginal value $\mu$ outside of the solution. To this end, denote with $S$ the output of the greedy algorithm with size $k$ on a ground set $X$; we define the largest normalized residual marginal $\mu \in [0,1]$ as the solution of 
        \begin{equation}
        \label{def:mu}
            \max_{e \in X \setminus S} f(e \mid S) = \frac{\mu}{k}f(\OPT).
        \end{equation}
        Note, we can expect two opposite behaviours for $\mu$ at the boundary of $[0,1]$: if $\mu \approx 0$, then it means that all the value is contained in $S$, while $\mu \approx 1$ says that the function is ``nearly'' additive. In both these two extreme cases, it is natural to think that \greedy outperforms its standard $1 - \nicefrac 1e$ approximation guarantee. In particular, in \Cref{app:refined_greedy} we prove the following Lemma. 
        
        \begin{restatable}{lemma}{greedyrefined}
        \label{lem:refined_greedy}
            The following inequality holds: $ f(S) \ge \left(1 + \mu \ln \mu\right) \cdot f(\OPT).$
        \end{restatable}
        
        Note, worst-case value of $\mu$ is $e^{-1}$, for which the standard $1 - \nicefrac1{e}$ bound is obtained.

        The main result of this Section is provided in the following theorem. 
            \begin{theorem}
            \label{thm:poly-robust}
                \GreedyCertificate with parameters $\threshold = 0.84$ and $\ExtraSample = 0.1$ is $0.51$-addition-robust and runs in polynomial time. 
            \end{theorem}
            
            To prove the above result, we show that for any disjoint sets of elements $\Vnow$ and $\Vfuture$, and any $R \subseteq \Vfuture$, it holds that 
            \[
                \E{f(\FinalSol \cup R)} \ge 0.51 f(\OPT'),
            \]
            where $\OPT'$ is the size $\k$ optimal solution in $\Vnow \cup R$, and $\FinalSol$ is the (randomized) output of \GreedyCertificate on $\Vnow$.
            Before tackling directly the proof of \Cref{thm:poly-robust}, we introduce some notation and provide some preliminary results. It is possible that the second stage of \GreedyCertificate does not fill all the $(1+\ExtraSample)\k$ elements, so we define $\ExtraSample'$ as the solution of 
            \(
                |{\AugmentedSol}| = (1+\ExtraSample')\k.
            \)
        %\morteza{Given the $\e$ fraction loss in selection process, we need to carry some $O(\e)$ terms in the analysis. They will appear in multiple places and we need to make a decision what is the minimal set of changes to make the proof correct. I will mention $O(\e)$ here and there so we can finalize it later.}
            \begin{lemma}[Lower bounds on $\FinalSol \cup R$]
            \label{lem:AcupR}
                We have the following inequalities: 
                \begin{itemize}
                    \item[(i)] $(1+\threshold \ExtraSample')f(\GreedySol) \le f(\AugmentedSol \cup R)$
                    \item[(ii)] $f(\AugmentedSol \cup R) \le (1+\ExtraSample') \E{f(\FinalSol \cup R)}  $
                    %\morteza{Depending on how we fix the capacity constraints in the algorithm, we may need to add an $O(\e)$ term here.}
                    \item[(iii)] $ \tfrac{1}{\ExtraSample'+1}f(\AugmentedSol \cup R) + \tfrac{\ExtraSample'}{1+\ExtraSample'} f(R) \le \E{f(\FinalSol \cup R)} $
                \end{itemize}
            \end{lemma}
            \begin{proof}
                We start by proving the first result. $\AugmentedSol$ has been augmented with elements $\{a_1,a_2, \dots a_{\ExtraSample'\k}\}$ respecting the threshold condition in the if statement (line~\ref{line:gamma_condition}), therefore we have the following:
                \begin{align*}
                    f(\AugmentedSol \cup R) &\ge f(\AugmentedSol) = f(\GreedySol) + \sum_{i = 1}^{\ExtraSample'\k} f(a_i|\GreedySol \cup \{a_1, \dots, a_{i-1}\})\ge (1+\threshold \ExtraSample')f(\GreedySol),
                \end{align*}
              where the last inequality is given by the fact that an element is added in the augmented solution only if its current marginal contribution is at least $\nicefrac \threshold{\k} f(\GreedySol).$
              
              Consider now the subsampling step, where $\k$ elements of $\AugmentedSol= \{a_1, \dots, a_{(1+\ExtraSample')\k}\}$ are drawn uniformly at random. We have the following chain of inequalities: 
              \begin{align*}
                  \E{f(\FinalSol \cup R)}\! &=\! f(R) +\!\!\!\! \sum_{i=1}^{(1+\ExtraSample')\k}\!\!\!\!\E{\ind{a_i \in \FinalSol} f(a_i| R \cup (\FinalSol \cap \{a_1,\dots, a_{i-1}\})} \tag{by  telescopic argument}\\
                  &\ge\!
             f(R) + \sum_{i=1}^{(1+\ExtraSample')\k}\P{a_i \in \FinalSol} f(a_i|  R \cup \{a_1,\dots, a_{i-1}\}) \tag{by submodularity} \\
             &\ge\! f(R) + \frac{1}{1+\ExtraSample'}f(\AugmentedSol |R). \tag{by design} \end{align*}
             The inequality in the statement of point (ii) follows by non-negativity of $f(R)$ and the fact that $\ExtraSample' \ge 0$. Starting for the last inequality, it is also easy how to derive the inequality in point (iii): it suffices to decompose $f(R)$ according to $\nicefrac{1}{1+\ExtraSample'}$ and $\nicefrac{\ExtraSample'}{1+\ExtraSample'}$.
            \end{proof}
            
            A further step in our analysis is provided by upper bounding the value of $\OPT'$ if the augmentation procedure stops before filling the extra budget. This provides a certificate that all the remaining elements have low value with respect to $f(\GreedySol).$  
            \begin{lemma}[Upper bound on $\OPT'$ when $\ExtraSample' < \ExtraSample$]
            \label{lem:delta>delta'}
                If $\ExtraSample' < \ExtraSample$, then the following inequality holds: 
                \[
                    f(\OPT') \le (1+\ExtraSample)\left(1+  \frac{\threshold}{1+\threshold \ExtraSample}
      \right) \E{f(\FinalSol\cup R)}
                \]
            \end{lemma}
            \begin{proof}
                We have the following chain of inequalities:
                \begin{align}
       f(\OPT') &\le f(\OPT' \cup \AugmentedSol \cup R)\tag{monotonicity}\\
       &\le f(\AugmentedSol \cup R) + \threshold f(\GreedySol) \tag{$ \forall x \in \OPT' \setminus \AugmentedSol$ it holds $f(x|\AugmentedSol \cup R) \le f(x|\AugmentedSol)\le \tfrac{\threshold}{\k} f(\GreedySol)$}\\
       &\le f(\AugmentedSol \cup R) \left[1 + \frac{\threshold}{1+\threshold \ExtraSample'}\right] \tag{By point (i) of \Cref{lem:AcupR}}\\
       &\le (1+\ExtraSample')\left(1+  \frac{\threshold}{1+\threshold \ExtraSample'}
      \right) \E{f(\FinalSol\cup R)}\tag{By point (ii) of \Cref{lem:AcupR}}\\
      &\le (1+\ExtraSample)\left(1+  \frac{\threshold}{1+\threshold \ExtraSample}
      \right) \E{f(\FinalSol\cup R)},
   \end{align}
   where in the last inequality we used the fact that $(1+x)(1+\tfrac{\threshold}{1+\threshold x})$ is monotonically increasing as a single variable function of $x$ for any parameter $\threshold \in [0,1]$\footnote{To see this last fact more clearly, it is enough to conveniently manipulate the expression: $(1+x)(1+\tfrac{\threshold}{1+\threshold x}) = 2 + x - \tfrac{1 - \threshold}{1+\threshold x}$. It is clear that as $x$ increases, the overall function also increases.}.
    \end{proof}
    The last preliminary step consists in upper bounding the value of the optimal solution $\OPT'$ of cardinality $\k$ in $\Vfuture$. To this end, denote with $\oldOPT$ the optimal solution with cardinality $\k$ in $\Vnow$, and recall the definition of $\mu$ as in Equation~\ref{def:mu} with cardinality $\k$. We have the following Lemma 
    \begin{lemma}[Upper bounds on $\OPT'$]
    \label{lem:upper_opt}
        The following inequalities hold true:
        \begin{itemize}
            \item[(i)] $f(\GreedySol) \ge (1+\mu \ln \mu) f(\oldOPT)$
            \item[(ii)] $f(\OPT') \le f(\oldOPT) + f(R)$
            \item[(iii)] $f(\OPT') \le f(\AugmentedSol \cup R) + \mu f(\oldOPT)$
        \end{itemize}
    \end{lemma}
    \begin{proof}
    The first point follows from \Cref{lem:refined_greedy} on $\Vnow$, with cardinality $\k$. For point (ii) we have the following:
    \begin{align*}
        f(\OPT') &\le f(\OPT' \cap R) + f(\OPT' \setminus R) \tag{By submodularity}\\
        &\le f(R) + f(\OPT' \cap \Vnow) \tag{By monotonicity}\\
        &\le f(R) + f(\oldOPT),
    \end{align*}
    where the last inequality follows by optimality of $\oldOPT$ in $\Vnow.$ To prove the last point, we need a different decomposition: 
    \begin{align*}
        f(\OPT') &\le f(\AugmentedSol \cup R) + f(\OPT'|\AugmentedSol \cup R) \tag{By monotonicity}\\
        &\le f(\AugmentedSol \cup R) + \sum_{x \in \OPT' \setminus  R}f(x|\GreedySol) \tag{By submodularity, as  $\GreedySol \subseteq \AugmentedSol \cup R$}\\
        &\le f(\AugmentedSol \cup R) + \mu f(\oldOPT),
    \end{align*}
    where the last inequality follows by definition of $\mu.$
    \end{proof}
    \begin{proof}[Proof of \Cref{thm:poly-robust}]
        We now address the case not covered by \Cref{lem:delta>delta'} (i.e., $\ExtraSample = \ExtraSample'$), and bound the corresponding approximation ratio. In particular, using the lower bound on $f(\FinalSol \cup R)$ as in point (iii) of \Cref{lem:AcupR}, and the two upper bounds on $\OPT'$ as in points (ii) and (iii) of \Cref{lem:upper_opt}, we have
        \begin{equation}
            \label{eq:competitive}
                \frac{f(\OPT')}{\E{f(\FinalSol \cup R)}} \le (1+\ExtraSample) \frac{\min\{f(\oldOPT) + f(R),f(\AugmentedSol \cup R) + \mu f(\oldOPT)\}}{f(\AugmentedSol \cup R) + \ExtraSample f(R)}
        \end{equation}
        We want to find the maximum value of this upper bound. This ratio can be seen as a function of  four variables: $f(\oldOPT), f(\AugmentedSol \cup R), f(R)$ and $\mu$. We set the parameter $\ExtraSample$ so it is not a variable. 
        
        As a first step, we replace $f(R)$ with a generic variable $x\ge 0$, and find its worst possible value, namely the value that maximizes the right side. To this end, it is convenient to decompose Inequality~\ref{eq:competitive} via the two terms that make up the $\min$: 
        \[
            \frac{f(\OPT')}{\E{f(\FinalSol \cup R)}} \le \max_{x \ge 0} \min\{g(x),h(x)\},
        \]
        where we introduced 
        \[
            g(x) =  \frac{(1+\ExtraSample)(f(\oldOPT) + x)}{f(\AugmentedSol \cup R) + \ExtraSample \cdot x}, \quad h(x) = \frac{(1+\ExtraSample)(f(\AugmentedSol \cup R) + \mu f(\oldOPT))}{f(\AugmentedSol \cup R) + \ExtraSample \cdot x}
        \]
        \begin{claim}
            As long as $\ExtraSample < \nicefrac {(e-1)}e$, the function $g$ is continuous and monotonically increasing for $x \ge 0$. Moreover, 
            \[
                g(0) =  \frac{(1+\ExtraSample) \cdot f(\oldOPT)}{f(\AugmentedSol \cup R)}, \quad \lim_{x \to \infty} g(x) = 1 + \frac{1}{\ExtraSample}
            \]
        \end{claim}
        \begin{proof}[Proof of the Claim]
            The only non-trivial result concerns the increasing nature of $g$. Consider the derivative of $g$:
            \[
                g'(x) = \frac{(1+\ExtraSample)(f(\AugmentedSol \cup R) -\ExtraSample f(\oldOPT))}{\left[f(\AugmentedSol \cup R) + \ExtraSample \cdot  x\right]^2}.
            \]
            Now, consider the numerator of the derivative. It is strictly positive:
            \begin{align*}
                f(\AugmentedSol \cup R) -\ExtraSample f(\oldOPT) &\ge  f(\GreedySol) -\ExtraSample f(\oldOPT) \tag{By monotonicity}\\
                &\ge (1+\mu \ln \mu - \ExtraSample) f(\oldOPT). \tag{By \Cref{lem:refined_greedy}}  
            \end{align*}
            Now, the last term is strictly positive  for all $\mu \in [0,1]$, as long as $\ExtraSample < 1 - \nicefrac 1e$.
        \end{proof}
        \begin{claim}
            Function $h$ is continuous and monotonically decreasing for $x \ge 0$. Moreover, 
            \[
                h(0) = \frac{(1+\ExtraSample)(f(\AugmentedSol \cup R) + \mu f(\oldOPT))}{f(\AugmentedSol \cup R)}, \text{ and } \lim_{x \to \infty} h(x) = 0
            \]
        \end{claim}
        The proof of the above Claim is a simple exercise as variable $x$ only appears in the denominator and also the numerator and the coefficient of $x$ are both non-negative. 
        We infer $h(0) \ge g(0)$ by point (iii) of \Cref{lem:upper_opt} and also noting $f(\OPT') \geq f(\oldOPT)$. For $x \to \infty$ it holds that $h$ is below $g$. Combining these two limit conditions with the above claims, we have that the worst choice of $x$ is the solution of $h(x) = g(x)$, i.e.:
        %Now, putting together the above two claims (note, $h(0) \ge g(0)$ and for $x \to \infty$ it holds that $h$ is below $g$), we have that the worst choice of $x$ is the solution of $h(x) = g(x)$, i.e.:
        \[
            x^\star= f(\AugmentedSol \cup R) - (1-\mu) f(\oldOPT).
        \]
        All in all, we have that 
        \begin{equation}
            \label{eq:last_step}
             \frac{f(\OPT)}{\E{f(\FinalSol \cup R)}} \! \le \! \frac{(1\!+\!\ExtraSample)(f(\AugmentedSol \cup R) \!+\! \mu f(\oldOPT))}{(1\!+\!\ExtraSample)f(\AugmentedSol \cup R) \!-\! \ExtraSample(1\!-\!\mu)f(\oldOPT)} \!\le\!\! \max_{y\ge(1+\threshold \ExtraSample)(1+\mu\ln\mu)}\! \frac{(1+\ExtraSample)(y + \mu)}{(1+\ExtraSample)y - \ExtraSample(1-\mu)},    
        \end{equation}
        where have denoted with $y$ the ratio $\nicefrac{f(\AugmentedSol \cup R)}{f(\oldOPT)}$. Note, the restriction on the domain of $y$ from $\mathbb R_{\ge0}$ to $[(1+\threshold \ExtraSample)(1+\mu\ln\mu), \infty)$ is crucial to go beyond the $\nicefrac 12$ barrier, and it is due to the improved analysis of the greedy algorithm. In particular, it is due to points (i) in \Cref{lem:AcupR} and \Cref{lem:upper_opt}:
        \[
            f(\AugmentedSol \cup R) \ge (1+\threshold \ExtraSample) f(\GreedySol) \ge  (1+\threshold \ExtraSample)(1+\mu \ln \mu) f(\oldOPT).
        \]
        To provide a uniform upper-bound on the right-hand-side term in Inequality~\ref{eq:last_step}, it is enough to argue that it is decreasing in $y$, as its derivative is strictly negative\footnote{The derivative is $-\tfrac{(1+\ExtraSample)(1+\mu)}{\left[(1+\ExtraSample)y - \ExtraSample(1-\mu)\right]^2}$}.  
        We can then plug $y = (1+\threshold \ExtraSample)(1+\mu\ln\mu)$ and have that
        \begin{align*}
            \frac{f(\OPT)}{\E{f(\FinalSol \cup R)}} &\le \frac{(1+\threshold \ExtraSample)(1+\mu\ln\mu) + \mu}{(1+\threshold \ExtraSample)(1+\mu\ln\mu) - \tfrac{\ExtraSample}{(1+\ExtraSample)}(1-\mu)}\\
            &\le \frac{(0.9225093 \mu + \mu \ln \mu + 1)}{(0.0838644 \mu + \mu \ln \mu + 0.916135)}\tag{By setting $\threshold = 0.84$ and $\ExtraSample = 0.1$}\\
            &\le \frac 1{0.5159}. \tag{For all $\mu \in [0,1]$}
        \end{align*}
        % To verify Google search this equivalent formula: 
        %   ((0.9225093*x + x*ln(x) + 1)) / ((0.0838644*x + x*ln(x) + 0.916135))
            To conclude the proof, it is enough to note that plugging $\threshold = 0.84$ and $\ExtraSample = 0.1$ into \Cref{lem:delta>delta'} yields a similar upper bound of $\nicefrac{1}{0.5121}$. 
        \end{proof}
        
        All in all, combining \Cref{thm:reduction} with \Cref{thm:poly-robust} we obtain the following result. 
        \begin{theorem}
            \label{thm:poly-time}
                For any $\e \in (0,1)$, there exists a randomized polynomial time algorithm that is $O(\nicefrac{1}{\e^2})$ consistent and provides a $(0.51 - O(\e))$-approximation of the dynamic optimum.
            \end{theorem}
        
        % \morteza{In the two cases of the above proof, we get the two approximation factor lower bounds of $0.5159$ and $0.5121$. If time permits, we can readjust the parameter settings and make the two bounds equal and improve the overall bound a little bit. Alternatively, we can add a sentence at the end of the proof that we did not optimize the approximation factor derived from this analysis beyond the second precision digit.}

%% file: sections/100-appendix.tex
\section{Appendix}

\subsection{A Tight Deterministic Addition Robust Algorithm}
\label{app:local}

    The \local algorithm, already proposed in the seminal paper by \citet{NemhauserWF78},  is $\nicefrac 12$-addition robust. The algorithm takes in input a precision parameter $\e$ and augments an initial greedy solution with local improvements, i.e., it adds elements to the solution $S$ as long as their marginal contribution with respect to $S$ is above $\nicefrac{f(S)}{\kappa}$ by a multiplicative margin $(1+\e)$. To make room for the local improvements, \local drops from the solution the element with smallest marginal contribution. For further details we refer to the pseudo-code. Note, in this section we adopt the shorthand $+$, respectively $-$, to denote the union, respectively set-difference.

\begin{algorithm}
\caption*{\local}
\begin{algorithmic}[1]
                \STATE \textbf{environment}: Submodular function $f$ on set $V$, cardinality constraint $\kappa$
            \STATE \textbf{input}: precision parameter $\e >0$
            \STATE Let $S$ be the greedy solution
            \WHILE{$\exists$ $x \in V: f(x|S) \ge \tfrac{(1+\e)}{\kappa} f(S)$}
            \STATE Let $x$ be such that $f(x|S) \ge
            \tfrac{(1+\e)}{\kappa} f(S)$
            \label{line:x}
            \STATE Let $y \in \argmin_{y \in S}\{f(y|S-y)\}$
            \label{line:y}
            \STATE $S \gets S + x - y$
            \ENDWHILE
            \STATE \textbf{Return $S$}
        \end{algorithmic}
        \end{algorithm}
        
        The crucial observation is that after each local improvement, the value of the solution increases by at least a $(1+\nicefrac{\e}{\kappa})$ multiplicative factor. In particular, fix any current solution $S$, and denote with $x$, respectively $y$, the element chosen in line~\ref{line:x}, respectively line~\ref{line:y} of \local. We have the following Claim. 
        \begin{lemma}
        \label{lem:local_improvement}
            The following inequality holds true:
            \[
                f(S+x-y) \ge \left(1+\frac{\e}{\kappa}\right) f(S).
            \]
        \end{lemma}
        \begin{proof}
            We have the following chain of inequalities:
            \begin{align*}
                f(x|S-y) &\ge f(x|S) \tag{By submodularity}\\
                &\ge \frac{1+\e}{\kappa} f(S) \tag{By the local improvement condition, see line~\ref{line:x}}.
            \end{align*}
            We can rearrange the above inequality and use that, by a simple averaging argument, the element $y$ defined in line~\ref{line:y} satisfies $\kappa \cdot f(S-y)\ge (\kappa-1) f(S).$ We have the following:
            \begin{align*}
                f(S+x-y) \ge \frac{1+\e}{\kappa} f(S) + f(S-y) \ge \left(1+\frac{\e}{\kappa}\right)f(S). 
            \end{align*}
            This concludes the proof.
        \end{proof}
        \Cref{lem:local_improvement} immediately implies that the algorithm converges after a finite number of steps. To see why this is the case, note that the initial greedy solution $S$ is such that $f(S) \ge (\nicefrac{(e-1)}{e}) f(\oldOPT)$, where $\oldOPT$ denotes the best $\kappa$ elements in $\Vnow$. 
        After $\ell$ local improvements, the current solution $S_{\ell}$ is such that 
        \[
            f(\oldOPT) \ge f(S_{\ell}) \ge \left(1+\frac{\e}{\kappa}\right)^{\ell}f(S_0) \ge \frac{e-1}{e}  \left(1+\frac{\e}{\kappa}\right)^{\ell} f(\oldOPT),
        \]
        where we denote with $S_0$ the initial greedy solution. This means that $\ell$ can be at most $O(\nicefrac{\kappa}{\e}).$
        
    By the design of the stopping condition of the while loop, it is immediate to argue that the output $S$ of \local is $(1+\e)$ stable {where we define $(1 + \e)$-stability as follows:} 
    \[
        f(x|S) < (1+\e) \frac{f(S)}{\kappa}, \forall x \in V.
    \]
    \begin{lemma}
    \label{lem:stable}
         \local outputs a $(1+\e)$-stable solution.
    \end{lemma}
    
    We have all the ingredients to prove the result. 
    
    \begin{theorem}
        For any choice of the precision parameter $\e \in (0,1)$, \local is $(\nicefrac{1}2-\e)$-addition robust.
    \end{theorem}
    \begin{proof}
        The proof is immediate: Let $R$ be any subset of $\Vfuture$, $S^\star$ be any set of cardinality at most $\kappa$ in $\Vnow \cup R$, and denote with $S$ the solution output by \local. Then we have the following:
        \begin{align*}
            f(S^\star) &\le f(S^\star \cup S \cup R) \\
            &\le f(S \cup R) + \sum_{x \in S^\star \setminus R} f(x|S) \tag{By submodularity}\\
            &\le f(S \cup R) + \sum_{x \in S^\star \setminus R} \frac{(1+\e)}{\kappa} f(S)\tag{By \Cref{lem:stable}} \\
            &\le (2+\e) f(S \cup R),
        \end{align*}
        where the last inequality follows by monotonicity and by noting that $S^{\star}$ contains at most $\kappa$ elements.
    \end{proof}

\subsection[Missing Proofs of Section 5]{Missing Proofs of \Cref{sec:expo}}
\label{sec:app-expo}

We prove the following:

\lemdisc*

\begin{proof}

     The lemma is easy when $f$ is a coverage function. Indeed, let $Y$ be the underlying universe.  We define the extension $\hat f$ and the set $\mathcal{R}$ by letting $\mathcal{R}$ contain an element for each subset of $Y$. In this way we have an ``equivalent'' element $r\in \mathcal{R}$ for each (unknown) $R \subseteq V_{\textsf{future}}$ in that $\hat f(S \cup \{r\}) = f(S \cup R)$ for every $S\subseteq V_{\textsf{now}}$. The properties of the lemma thus holds in this case.

     Let us now consider the case when $f$ is a general non-negative. monotone, submodular function. 
    Let $\delta = \frac{\e}{4} \max_{V' \subseteq V_{\textsf{now}}: |V'| = \kappa} f(V')$ be the discretization parameter and let $m = \lceil 2 \nicefrac{f(V_{\textsf{now}})}{(\e \delta)} \rceil$.
    In our definition of $\mathcal{R}$ we consider the set of vectors $\mathcal{V} = \{0, \delta, 2\delta, \ldots, m \cdot \delta\}^{2^{|V_{\textsf{now}}|}}$. In words, each vector in this family has $2^{|V_{\textsf{now}}|}$ dimensions and each dimension takes a value between $0$ and $m\delta$ that is a multiple of $\delta$. It will be convenient to index the dimensions by the subsets $S \subseteq V_{\textsf{now}}$.
    Let us provide some intuition behind these vectors.  %For each possible future scenario $R$, we intuitively have a vector $v\in \mathcal{V}$ so that $v_S \approx f(S \cup R)$ for every $S \subseteq \Vnow$, and we would like to have an element $r(v) \in \mathcal{R}$ so that $\hat f(S \cup r(v)) \approx v_S \approx f(S \cup R)$, which will allow us to "replace"  replace the unknown scenario $R$ with the almost identical known element $r(v)$. So the value $v_S$ intuitively corresponds to the ``guessed'' value of $\hat f(S \cup r(v))$ for $S \subseteq \Vnow$.
    %
    %Our extension Each of these vectors represent a scenario of the future denoted by one of the elements in $\mathcal{R}$. To define this scenario, we would need to define how the future set interacts with all subsets of the set $V_{\textsf{now}}$. Each dimension of the vector $v \in  \mathcal{V}$ represents the value of the extension $\hat f(\cdot)$ when the future elements are added to a subset of $V_{\textsf{now}}$. For instance if $|V_{\textsf{now}}| = 2$, then $v \in  \mathcal{V}$ has four dimensions corresponding to the four subsets of $V_{\textsf{now}}$ and how they interact with the new element, 
%    i.e. how much value the union of the new element and the corresponding subset has.
%    For a $v\in \mathcal{V}$ and a set $S\subseteq V_{\textsf{now}}$, we denote with $v_S$ the entry of the vector $v$ corresponding to the ``guessed'' value of $S$ union the new element.
     Our set $\mathcal{R}$ will contain an element $r(v)$ for every vector $v\in \mathcal{V}$, and the intuition is that we would like the extension $\hat f$ to assign values $\hat f( S\cup \{r(v)\}) = v_S$ to every set that contains $r(v)$. By the fine discretization of the vectors in $\mathcal{V}$ (each dimension is a multiple of $\delta$), we have,  for every future scenario $R$ (with $f(R) \leq f(\Vnow)/\e$ ),  a vector $v\in \mathcal{V}$ so that $|v_S - f(S \cup R)| \leq \delta$ for every set $S \subseteq V_{\textsf{now}}$. Thus our "ideal" selection of $\hat f$ would ensure $\hat f(S \cup r(v)) \approx f(S \cup R)$ for every $S \subseteq \Vnow$, which would allow us to replace the unkonown future $R$ with the almost identical known element $r(v)$. There is, however, one more technicality: the guessed values $v_S$ may not correspond to a non-negative monotone submodular function. We address this by finding the values $\hat f(S \cup r(v))$  that minimizes the largest deviation $\max_{S\subseteq \Vnow} |f(S\cup \{r(v)\}) - v_S|$  subject to $\hat f$ being a non-negative monotone submodular function when restricted to subsets of $\Vnow \cup \{r(v)\}$. The largest deviation  $|\hat f(S \cup \{r(v)\})- v_S|$ is measured by the variable $\textsf{error}$ which one can see is bounded by $\delta$ for any $r(v)$ that is ``close'' to a possible future scenario $R$. Indeed, only the vectors $v\in \mathcal{V}$ for which $\textsf{error} \leq \delta$ will be important in the proof of the lemma.

%resolved    \morteza{Following, we use $r(v)$ to represent all potential submodular functions $f$ such that $f(S \cup r(v))$ is close to $v_S$. First of all, I am not sure if for every $r(v)$ a "nearby" submodular function $f$ exists  to realize it. But in that case, the LP won't necessarily find a solution with $error \leq \delta$. The important point is that this LP helps us do a many to one mapping from potential functions $f$ to extension functions $\hat{f}$. For the rest of the proof, we show any possible function $f$ is represented by a nearby $\hat{f}$ that we find in at least one of the LP solutions. That is sufficient for the proof. How does this sound? }

    Formally, we now construct $\mathcal{R}$ as follows. For each vector $v\in \mathcal{V}$, we have an element $r(v)$ in $\mathcal{R}$, and  
    the values of $\hat f$ on sets including element $r(v)$ are obtained by solving the following linear program\footnote{Recall that we can linearize the inequality  $|\hat f(S \cup \{r(v)\}) - v_S|  \leq \textsf{error}$ by replacing it with the two inequalities $\hat f(S \cup \{r(v)\}) - v_S \leq \textsf{error}$ and $v_S - \hat f(S \cup \{r(v)\}) \leq \textsf{error}$.} that has a variable $\hat f(S \cup \{r(v)\})$ for every $S \subseteq V_{\textsf{now}}$:
    \begin{align*}
        \mbox{Minimize \quad} & \textsf{error} \\[2mm]
        \hat f(S) + \hat f(T) & \geq \hat f(S\cap T) +\hat f(S\cup T) \qquad  \quad \mbox{for all $S,T \subseteq V_{\textsf{now}} \cup \{r(v)\}$ with $r(v) \in S\cup T$,}  \\
         \hat f(S) &\leq \hat f(T)  \qquad \qquad \qquad \qquad \qquad \mbox{for all $S\subseteq T \subseteq V_{\textsf{now}} \cup \{r(v)\}$ with $r(v) \in T$,} \\
         \hat f(S) & \geq 0 \qquad \qquad \qquad \qquad \qquad \qquad \mbox{for all $S \subseteq V_{\textsf{now}} \cup \{r(v)\}$ with $r(v) \in S$,} \\
         |\hat f(S \cup \{r(v)\}) - v_S| & \leq \textsf{error} \qquad \qquad \qquad \qquad \qquad \mbox{ for all $S \subseteq V_{\textsf{now}}$.} 
    \end{align*}

    We emphasize that the only variables are of the type $\hat f(S)$ when $S$ contains $r(v)$ and otherwise, if $S$ does not contain $r(v)$, the value is already defined because $\hat f$ is an extension of $f$ that is already defined on all the subset of $V_{\textsf{now}}$. 
    The first set of constraints ensures that $\hat f$ is submodular on $V_{\textsf{now}} \cup \{r(v)\}$, the second ensures monotonicity, and the third non-negativity. We thus have that the set $\mathcal{R}$ and the valuations of $\hat f$ decided as above satisfies the first property of the lemma. We remark that the above linear program assigns values to all sets $\hat f(S \cup \{r(v)\})$ with $S \subseteq V_{\textsf{now}}$. We have thus extended $f$ to $\hat f$ by defining values for every set $\hat f(S \cup \{r\})$ with $r\in \mathcal{R}$ and $S \subseteq V_\textsf{now}$ so that the first property is satisfied. Also note that the value of subsets that contain two or more elements of $\mathcal{R}$ are irrelevant to the statements of the lemma and we simply set the value of $\hat f$ on such subsets to $0$.
    
    Having described $\mathcal{R}$ and the extension of $\hat f$ to these elements, we now continue to argue the second property. To this end, assume 
    \begin{align}
    \EO_{A \sim \DistributionA }[{\hat f(A \cup \{r\})}]  \geq \alpha \cdot \max_{V' \subseteq V_{\textsf{now}}: |V'| = \kappa} \hat f(V' \cup \{r\}) \qquad \mbox{ for every $r \in \mathcal{R}$}
    \label{eq:assumption}
    \end{align}
    We now fix a set $R \subseteq V_{\textsf{future}}$ and wish to argue that 
    \begin{align*}
    \EO_{A \sim \DistributionA }[{f(A \cup R)}]  \geq (\alpha- \e) \cdot \max_{V' \subseteq V_{\textsf{now}}: |V'| = \kappa} f(V' \cup R)\,.
    \end{align*}
    First note that if $f(R) \geq \nicefrac{f(V_{\textsf{now}})}\e$ then the above holds immediately as
    \begin{gather*}
        f(R) \geq  \tfrac{1}{\e}f(V_{\textsf{now}}) \qquad \mbox{and} \qquad \max_{V' \subseteq V_{\textsf{now}}: |V'| = \kappa} f(V' \cup R) \leq f(V_{\textsf{now}} \cup R) \leq f(R) + f(V_{\textsf{now}})\,.
    \end{gather*}
    In other words, in that case, we have that 
    \[
        (1-\e) \max_{V' \subseteq V_{\textsf{now}}: |V'| = \kappa} f(V' \cup R) \leq (1-\e) f(V_{\textsf{now}} \cup R)  \leq f(R).
    \]
    
    Let us now focus on the more interesting case when $f(R) < \nicefrac{f(V_{\textsf{now}})}\e$. In that case 
    \[
        f(R \cup S)\leq f(V_{\textsf{now}})+\tfrac{1}{\e}f(V_{\textsf{now}}) \leq \tfrac{2}{\e}f(V_{\textsf{now}}),
    \]
    for any $S \subseteq V_{\textsf{now}}$. This implies, by the selection of $m$, that there is a vector $v\in \mathcal{V}$ 
    so that 
    \begin{align*}
        | f(S \cup R) - v_S | \leq \delta  \qquad \mbox{for every $S \subseteq V_{\textsf{now}}$}.
    \end{align*}
    Moreover, $|\hat f(S \cup \{r(v)\}) - v_S| \leq \textsf{error}$ where \textsf{error} is upper bounded by $\delta$ as setting $\hat f(S \cup \{r(v)\}) = f(S\cup R)$  would be one feasible solution to the linear program that minimizes $\textsf{error}$. We therefore have
    \begin{align*}
        | f(S \cup R) - \hat f(S \cup \{r(v)\})| \leq 2\delta  \qquad \mbox{for every $S \subseteq V_{\textsf{now}}$}.
    \end{align*}
    Thus~\eqref{eq:assumption} with $r=r(v)$ implies 
    \begin{align*}
    \EO_{A \sim \DistributionA }[{f(A \cup R)}] + 2\delta  \geq \alpha \cdot \left(\max_{V' \subseteq V_{\textsf{now}}: |V'| = \kappa} f(V' \cup R) - 2\delta \right)
    \geq \alpha \cdot \left(\max_{V' \subseteq V_{\textsf{now}}: |V'| = \kappa} f(V' \cup R)\right) - 2\delta\,.
    \end{align*}
    Property two of the lemma now holds because of the selection of $\delta = \frac{\e}{4} \max_{V' \subseteq V_{\textsf{now}}: |V'| = \kappa} f(V')$ which by monotonicity is at most $\frac{\e}{4} \max_{V' \subseteq V_{\textsf{now}}: |V'| = \kappa} f(V' \cup R)$.

\end{proof}

\paragraph{Functions $g_i$ are monotone and submodular.}
We show that the function $g_i(.)$ defined in \Cref{sec:expo} for $1 \leq i \leq n$ is monotone, submodular. To that end, let us recall its definition.

\begin{itemize}
    \item $g_i(\emptyset) = 0$.
    \item $g_i(a_j) = f(r) = 1$ for $1 \leq j \leq n$.
    \item $g_i(\{a_j, r\}) = \nicefrac 43$ for $1 \leq j \leq n$ and $i \neq j$.
    \item $g_i(\{a_j, a_k\}) = \nicefrac 53$ for $1 \leq j,k \leq n$ and $j \neq k$. %\ola{I updated this from 1.8 to 1.6. I think it was wrong previously.}
    %\silvio{Can an algorithm always play $\{a_i, a_j\}$ and get 0.8?}
    \item  $g_i(\{a_j , a_k , r\})  = \nicefrac 53$, for $1 \leq j,k \leq n$ and $j \neq k$, $i \neq j$, $i \neq k$. %\ola{Also updated this case}
%    \item $f(\{a_i , r_i , a_j\} = f(\{a_i , r_i , r_j\} = 2$, for $1 \leq i,j \leq n$ and $i \neq j$.
    \item $g_i(S) = 2$, for any other $S \subseteq X_f$, i.e., if $S$ satisfies one of the following: $|S\cap \{a_1, \ldots, a_n\}| \geq 3$ or $\{a_i, r\} \subseteq S$.
\end{itemize}

\begin{lemma}
For any $1 \leq i \leq m$, the function $g_i$ is monotone and submodular.
\end{lemma}
\begin{proof}
We start by showing that $g_i$ is  monotone, i.e., that $g_i(e \mid S) \geq 0$ for all subsets $S$ and element $e$ not in $S$. We have the following cases, according to the cardinality of $S$ and whether $a_i$ and $r$ belong to $S$ or not. 
\begin{itemize}
    \item If $S = \emptyset$, then $g_i(e \mid S) = 1$ for all possible $e$.
    \item If $|S| = 1$, then $g_i(e \mid S) = g_i(\{e\} \cup S) -  g_i(S) \geq \nicefrac 43 - 1 > 0$.
    \item If $|S| \geq 2$ and $\{a_i, r\} \subseteq S$, then $g_i(e \mid S) = 0$ for all possible $e$.
    \item If $|S| = 2$ but $\{a_i, r\} \neq S$, then $g_i(e \mid S) \geq \nicefrac 53 - \nicefrac 53 \geq 0$.
    \item If $|S| \geq 3$, then $g_i(e \mid S) \geq 2 - 2 \geq 0$.
\end{itemize}

    We conclude the lemma by arguing about the submodularity of $g_i.$ In particular, we show that for any subsets $S_1 \subset S_2 \subseteq \{a_1, \cdots, a_n, r\}$ with $|S_1| + 1 = |S_2|$, and $e\notin S_2$, it holds that $g_i(e \mid S_1) \geq g_i(e \mid S_2)$. We have the following case analysis.

\begin{enumerate}
    \item $S_1 = \emptyset$. We have $g_i(e \mid S_1) = 1$ and $g_1(e \mid S_2) \leq 1$ so $g_i(e \mid S_1) \geq g_i(e \mid S_2)$.
    \item $S_1 = \{a_i\}$. We have the following sub-cases:
        \begin{itemize}
            \item $S_2 = \{a_i, a_j\}$ and $e = r$, for distinct $i$ and $j$. Then $g_i(e \mid S_1) = 2 - 1 =1$ and $g_1(e \mid S_2) = 2-\nicefrac 53 = \nicefrac 13$.
            \item $S_2 = \{a_i, a_j\}$ and $e \neq r$, for distinct $i$ and $j$. Then $g_i(e \mid S_1) = \nicefrac 53-1 = \nicefrac 23$ and $g_i(e \mid S_2) = 2-\nicefrac 53 = \nicefrac 13$.
            \item $S_2 = \{a_i, r\}$ and $e \neq r$. Then $g_i(e \mid S_1) = \nicefrac 53-1 = \nicefrac 23$ and $g_i(e \mid S_2) = 2 - 2 = 0$.
        \end{itemize}
    \item $S_1 = \{a_j\}$, for $j\neq i$. We have the following sub-cases:
        \begin{itemize}
            \item $S_2 = \{a_j, a_k\}$, for $i$, $j$ and $k$ distinct, and $e = r$. Then $g_i(e \mid S_1) = 
            \nicefrac 43 - 1 = \nicefrac 13$ and $g_i(e \mid S_2) = \nicefrac 53-\nicefrac 53 = 0$.
            \item $S_2 = \{a_j, a_k\}$, for $i$, $j$ and $k$ distinct, and $e \neq r$. Then $g_i(e \mid S_1) = \nicefrac 53-1 = \nicefrac 23$ and $g_i(e \mid S_2) = 2-\nicefrac 53 = \nicefrac 13$.
            \item $S_2 = \{a_j, r\}$ and $e = a_i$. Then $g_i(e \mid S_1) = \nicefrac 53 - 1 = \nicefrac 23$ and $g_i(e \mid S_2) = 2-\nicefrac 43 = \nicefrac 23$.
            \item $S_2 = \{a_j, r\}$ and $e \neq a_i$. Then $g_i(e \mid S_1) = \nicefrac 53-1 = \nicefrac 23$ and $g_i(e \mid S_2) = \nicefrac 53-\nicefrac 43 = \nicefrac 13$.
            \item $S_2 = \{a_j, a_i\}$ and $e = r$. Then $g_i(e \mid S_1) = \nicefrac 43-1 = \nicefrac 13$ and $g_i(e \mid S_2) = \nicefrac 53-\nicefrac 53 = 0$.
            \item $S_2 = \{a_j, a_i\}$ and $e \neq r$. Then $g_i(e \mid S_1) = \nicefrac 53-1 = \nicefrac 23$ and $g_i(e \mid S_2) = 2-\nicefrac 53 = \nicefrac 13$.
        \end{itemize}
    \item $S_1 = \{r\}$. We have the following sub-cases
        \begin{itemize}
            \item $S_2 = \{r, a_j\}$, for $j \neq i$ and $e = a_1$. Then $g_i(e \mid S_1) = 2 - 1 = 1$ and $g_i(e \mid S_2) = 2-\nicefrac 43 = \nicefrac 23$.
            \item $S_2 = \{r, a_j\}$, for $j \neq i$ and $e \neq  a_1$. Then $g_i(e \mid S_1) = \nicefrac 43 - 1 = \nicefrac 13$ and $g_i(e \mid S_2) = \nicefrac 53-\nicefrac 43 = \nicefrac 13$.
            \item $S_2 = \{r, a_i\}$. Then $g_i(e \mid S_1) = \nicefrac 43 - 1 = \nicefrac 13$ and $g_i(e \mid S_2) = 2-2 = 0$.
        \end{itemize}
    \item $S_1 = \{a_i, r\}$. Then $g_i(e \mid S_1) = 0$ and $g_i(e \mid S_2) = 0$ for all possible $e, S_2$.
    \item $S_1 = \{a_i, a_j\}$. We have the following sub-cases
        \begin{itemize}
            \item $S_2 = \{a_i, a_j, a_k\}$, then $g_i(e \mid S_1) \geq 0$ and $g_i(e \mid S_2) = 2 - 2 = 0$, for all possible $e$.
            \item $S_2 = \{a_i, a_j, r\}$, then $g_i(e \mid S_1) = 2- 2= 0$ and $g_i(e \mid S_2) = 2 - 2 = 0$, for all possible $e$.
        \end{itemize}
    \item $S_1 = \{a_j, a_k\}$, for $i$, $j$ and $k$ distinct. We have the following sub-cases
        \begin{itemize}
            \item $S_2 = \{a_j, a_k, a_i\}$, then $g_i(e \mid S_1) \geq 0$ and $g_i(e \mid S_2) = 2 - 2 = 0$, for all possible $e$.
            \item $S_2 = \{a_j, a_k, r\}$ and $e = a_i$, then $g_i(e \mid S_1) \geq 2- \nicefrac 53 = \nicefrac 13$ and $g_i(e \mid S_2) = 2 - \nicefrac 53 = \nicefrac 13$.
            \item $S_2 = \{a_j, a_k, a_\ell\}$, for $i$, $j$, $k$, and $\ell$ distinct, then $g_i(e \mid S_1) \geq 0$ and $g_i(e \mid S_2) = 0$ for all possible $e$.
        \end{itemize}
    \item $S_1 = \{a_j, r\}$, for $j \neq i$. We have the following sub-cases:
        \begin{itemize}
            \item $S_2 = \{a_j, r, a_i\}$, then $g_i(e \mid S_1) \geq 0$ and $g_i(e \mid S_2) = 2 - 2 = 0$, for all possible $e$.
            \item $S_2 = \{a_j, r, a_k\}$ and $e = a_i$, then $g_i(e \mid S_1) = 2- \nicefrac 43 = \nicefrac 23$ and $g_i(e \mid S_2) = 2 - \nicefrac 53 = \nicefrac 13$.
            \item $S_2 = \{a_j, r, a_k\}$ and $e \neq a_i$, then $g_i(e \mid S_1) = \nicefrac 53 - \nicefrac 43 = \nicefrac 13$ and $g_i(e \mid S_2) = 2 - \nicefrac 53 = \nicefrac 13$.
        \end{itemize}
    \item $|S_1| \geq 3$, then $|S_2| \geq 4$, so $g_1(e \mid S_2) = 0$ for all possible $e$.
\end{enumerate}

\end{proof}

 \subsection[Missing Proofs of Section 6]{Missing Proofs of \Cref{sec:poly-time}}
\label{app:refined_greedy}

In this Section, we provide a finer analysis of the greedy algorithm by \citet{NemhauserWF78}. The setting is the standard monotone submodular maximization problem with cardinality $k$. In particular, we denote with $S$ the greedy output on a ground set $X$, with $\OPT$ an optimal solution on $X$, while $\mu$ is defined as follows \[
    \max_{e \in X \setminus S} f(e|S) = \frac{\mu}{k} f(\OPT).
\]
\greedyrefined*
            \begin{proof}[Proof of \Cref{lem:refined_greedy}]
                {If $\mu < \nicefrac 1e$, then the argument is simple: 
                \begin{align*}
                    f(\OPT) &\le f(\OPT \cup S) \tag{By monotonicity}\\
                    &\le f(S) + \sum_{x \in \OPT \setminus S} f(x|S)\tag{By submodularity}\\
                    &\le f(S) + \mu f(\OPT) \tag{By definition of $\mu$ and $|\OPT|\le k$}\\
                    &\le f(S) - \mu \ln \mu f(\OPT), \tag{Because $\mu < \nicefrac 1e$ implies $ 1 < - \ln \mu$}
                \end{align*}    
                so that the statement holds by rearranging the terms in the inequality. Focus now on the remaining case, i.e., $\mu \ge \nicefrac 1e$;} we assume without loss of generality that the cardinality of the greedy solution $S$, and of the optimal solution $\OPT$ is maximal (i.e., $|S| = |\OPT| = k$). Sort the elements in $S$ according to the order in which they have been added to the solution $S = \{s_1, \dots, s_k\}$, and denote with $S_i$ the set of the first $i$ elements added to $S$: $S_i = \{s_1, \dots, s_i\}$. Let $S_0 = \emptyset$ to simplify the proof. We have the following Claim. 
            \begin{claim}
            \label{cl:induction}
                The following inequality holds for all $i = 1, \dots, k$:
                \[
                    f(\OPT) - f(S_i) \le \left( 1 - \frac 1k\right)^i f(\OPT)
                \]
            \end{claim}
            \begin{proof}[Proof of \Cref{cl:induction}]
                We prove the Claim by induction on $i$. The inequality is trivially true for $i=0$, so we assume by induction it holds for a generic $i-1$ and argue that it then holds for $i$. By monotonicity and submodularity, we have:
                \begin{equation}
                    \label{eq:monotonicity}
                    f(\OPT) \le f(\OPT \cup S_{i-1}) \le f(S_{i-1}) + \sum_{o \in \OPT \setminus S_{i-1}} f(o|S_{i-1}).
                \end{equation}
                Rearranging the terms in Inequality~\ref{eq:monotonicity}, it holds that 
                \begin{align}
                \nonumber
                    f(\OPT) - f(S_{i-1}) &\le \sum_{o \in \OPT \setminus S_{i-1}} f(o|S_{i-1})\\
                \nonumber
                    &\le k \cdot f(s_i | S_{i-1}) \tag{Greedy Property}\\
                \label{eq:O-S_i}   &= k f(S_i) - k f(S_{i-1}).
                \end{align}
                We can then finalize the proof of the Claim:
                \begin{align*}
                    f(\OPT) - f(S_i) &= f(\OPT) - f(S_{i-1}) - f(S_i | S_{i-1})\\
                    &\le \left(1 - \frac 1k\right)(f(\OPT) - f(S_{i-1})) \tag{By Inequality \ref{eq:O-S_i}}\\
                    &\le \left(1 - \frac 1k\right)^i f(\OPT),
                \end{align*}
                where the last inequality follows the inductive hypothesis and the fact that $f$ is non-negative.
            \end{proof}
            
            For each $S_i$, we know that all the remaining elements in $S$ have marginal contribution which is at least $(\nicefrac{\mu}{k})f(\OPT)$, by submodularity and definition of $\mu$. Combining this fact with \Cref{cl:induction}, we have the following inequality holding for all integer $i=1, \dots, k$:
            \begin{align}
                \nonumber
                    f(S) &\ge f(S_i) + \sum_{j=i+1}^k f(s_j|S_{j-1})\\
                \nonumber
                    &\ge f(S_i) + (k - i) \frac{\mu}{k} f(\OPT)\\
                \nonumber
                    &\ge\left[1 - \left(1 - \frac 1k\right)^i + (k - i) \frac{\mu}{k}\right]f(\OPT)\\
                \label{eq:integer_inequality}
                    &\ge\left[1 - e^{-\nicefrac ik} + (k - i) \frac{\mu}{k}\right]f(\OPT),
            \end{align}
            where the last inequality follows from the fact that {$(1-\nicefrac 1k)^i \le e^{-\nicefrac ik}$} for all $i$. 
            {This is true because $e^x \geq 1+x$ for any real $x$, applying $x = \nicefrac 1k$ in this case.}
            
            {We define $x^* = k \ln \nicefrac {1}\mu$ {which is smaller or equal than $k$ as long as $\mu \ge \nicefrac 1e$} There are two cases. If $x^*$ is integer}, 
            then we can plug it in Inequality~\ref{eq:integer_inequality}  by setting $i = x^*$ and obtain the inequality claimed in the statement. If $x^*$ is not integer, then we need something more. First, $x^*$ can be written as $x^* = i^* + \delta$, for some integer $i^*$ and $\delta \in (0,1)$. By \Cref{cl:induction} and Inequality~\ref{eq:integer_inequality}, we have that 
            \begin{equation}
            \label{eq:i^*_step}
                f(S_{i^*}) \ge \left(1 - e^{-\nicefrac {i^*}k}\right) f(\OPT). 
            \end{equation}
            Consider the element ($e_{i^*+1}$) added by \greedy when the current solution is $S_{i^*}$, 
            % \morteza{Shouldn't this be $S_{i^*}$ instead of $S_{i^*+1}$?}
            and divide it into two fractional elements $e'$ and $e''$ of weights $\delta$ and $(1-\delta)$ 
            %such that $f(e'|S) = \delta f(e_{i^*+1}|S_{i^*})$, and $f(e''|S) = (1-\delta) f(e_{i^*+1}|S_{i^*})$.
            {such that $f(e'|S_{i^*}) = \delta f(e_{i^*+1}|S_{i^*})$, and $f(e''|S_{i^*} \cup \{e'\}) = (1-\delta) f(e_{i^*+1}|S_{i^*})$. }
            
            % \morteza{I think the previous version of the above definition of marginal gains of $e'$ and $e''$ were not correct. The old version is commented out. Please check and verify.}
            
            We can imagine that \greedy first add $e'$ and then $e''$ to its current solution $S_{i^*}$. Let $\hat S = S_{i^*} + e'$, since $e'$ has the largest marginal density (the same as $e_{i^*+1}$, which is the greedy choice), we have the following:
            \[
                f(e_{i^*+1}|S_{i^*}) = \tfrac{1}{\delta}(f(\hat S) - f(S_{i^*})), 
            \]
            which implies (plugging it into Inequality~\ref{eq:O-S_i}) that 
            \[
                f(\hat S) \ge \left[1 - \left(1 - \frac 1k\right)^{i^*}\left(1 - \frac {\delta}k\right)  \right]f(\OPT) \ge \left[1 - e^{-\tfrac{i^*+\delta}k}\right]f(\OPT).
            \]
            After adding the fractional $e'$, the remaining (fractional) room in the solution is $k - i^* - \delta = k - x^*$, 
            therefore we have the following inequality:
            \[
                f(S) \ge f(\hat S) + (k - x^*) \frac{\mu}{k} f(\OPT) \ge \left[1 - e^{-\nicefrac {x^*}k} + (k - x^*) \frac{\mu}{k}\right]f(\OPT). 
            \]
            Plugging in the value of $x^*$ yields the desired bound.
        \end{proof}
     
%         %Note, the same result also yields for randomized algorithms against an adaptive adversary!